%% file: main.tex
\RequirePackage{afterpackage}
\AfterPackage{amsthm}{
  \RequirePackage{hyperref,cleveref}
}

\documentclass[a4paper,UKenglish]{lipics-v2019}
\nolinenumbers

\usepackage{xspace}
\usepackage{microtype} 
\usepackage{thmtools}
\usepackage{thm-restate}
\usepackage{mathtools}
\usepackage{bm}

\usepackage[inline]{enumitem}

\usepackage[normalem]{ulem}

\usepackage{wrapfig}

\usepackage{subdepth}

\usepackage{newfloat}
\DeclareFloatingEnvironment{algorithm}
\usepackage{algpseudocode}

\graphicspath{{./graphics/}} 

\title{Sequentiality of String-to-Context Transducers}

\author{Pierre-Alain Reynier}{Aix-Marseille Université, CNRS, LIS UMR 7020, France}
  {pierre-alain.reynier@lis-lab.fr}{}{}
\author{Didier Villevalois}{Aix-Marseille Université, CNRS, LIS UMR 7020, France}
  {didier.villevalois@lis-lab.fr}{}{}

\authorrunning{P.-A. Reynier and D. Villevalois}

\Copyright{Pierre-Alain Reynier and Didier Villevalois}

\ccsdesc[F.1.1]{Models of Computation}
\ccsdesc[F.4.3]{Formal Languages}


\usepackage{xcolor}

\usepackage{tikz}
\usetikzlibrary{arrows,shapes,automata, calc, chains, matrix, positioning, scopes}
\usetikzlibrary{arrows.meta}

\input{macros}

\input{figure/macros}

\begin{document}

\maketitle


\begin{abstract}
\input{000-abstract}
\end{abstract}


\section{Introduction}
\label{sec:introduction}
\input{010-introduction}


\section{Models}
\label{sec:preliminaries}
\input{020-preliminaries}


\section{Lipschitz and Twinning Properties}
\label{sec:lipschitz-twinning}
\input{030-lipschitz-twinning}


\section{Main Result}
\label{sec:main-result}
\input{040-main-result}


\section{Analysis of Loop Combinatorics}
\label{sec:combinatorics}
\input{050-combinatorics}


\section{Determinisation}
\label{sec:construction}
\input{060-construction}


\section{Decision}
\label{sec:decision}
\input{070-decision}


\section{Conclusion}
\label{sec:conclusion}
\input{080-conclusion}


\bibliography{concatenation-free}


\newpage

\appendix


\section{Proofs of \Cref{sec:preliminaries}: Models}
\input{120-annex-preliminaries}


\section{Proofs of \Cref{sec:combinatorics}: Analysis of Loop Combinatorics}
\input{150-annex-combinatorics}

\section{Proofs of \Cref{sec:construction}: Determinisation}
\input{160-annex-construction}

\section{Proofs of \Cref{sec:decision}: Decision}
\input{170-annex-decision}

\end{document}

%% file: macros.tex
\newcommand{\T}{\mathcal{T}}
\newcommand{\D}{\mathcal{D}}
\newcommand{\Hc}{\mathcal{R}}

\newcommand{\NFT}{S2S}

\newcommand{\StoC}{S2C}
\newcommand{\longStoC}{string-to-context transducer}

\newcommand{\QQ}{\ensuremath{|Q|^{|Q|}}}

\newcommand{\pf}[2]{{\cal F}(#1,#2)}
\newcommand{\Cont}[1]{{\cal C}_{#1}(\outA)}

\newcommand{\extract}{\textsf{extract}}

\newcommand{\word}{\textsf{word}}

\newcommand{\resp}{\emph{resp.}\ }
\newcommand{\intro}[1]{\textit{#1}}

\renewcommand{\leq}{\leqslant}
\renewcommand{\geq}{\geqslant}
\renewcommand{\le}{\leqslant}
\renewcommand{\ge}{\geqslant}

\renewcommand{\emptyset}{\varnothing}

\newcommand{\dom}{\textsf{dom}}
\newcommand{\N}{\mathbb{N}}
\newcommand{\Nplus}{\mathbb{N}_{> 0}}

\newcommand{\alphabet}{A}

\newcommand{\commentPA}[1]{}

\newcommand{\base}{M_\T}
\newcommand{\basep}{M_{\T'}}

\newcommand{\dist}{\ensuremath{\textsl{dist}_p}}
\newcommand{\distf}{\ensuremath{\textsl{dist}_f}}

\newcommand{\lcc}{\textsf{lcc}}

\newcommand{\lcp}{\ensuremath{\textsf{lcp}}}
\newcommand{\lcs}{\ensuremath{\textsf{lcs}}}
\newcommand{\lcf}{\ensuremath{\textsf{lcf}}}

\newcommand{\TP}[1]{\ensuremath{\textsl{TP}_{#1}}}

\newcommand{\inter}[1]{ \lbrack \! \lbrack #1 \rbrack \!\rbrack}

\newcommand{\select}[1][]{%
\ifthenelse{\equal{#1}{}}{ \textsf{select}_1}{ \textsf{select}_1(#1)}%
}
\newcommand{\selectt}[1][]{%
\ifthenelse{\equal{#1}{}}{ \textsf{select}_2}{ \textsf{select}_2(#1)}%
}

\newcommand{\tinit}{t_{\textsl{init}}}
\newcommand{\tfinal}{t_{\textsl{final}}}

\newcommand{\pto}{\hookrightarrow}
\newcommand{\eps}{\varepsilon}

\newcommand{\letter}[1]{\ensuremath{\textcolor{magenta}{#1}}}
\newcommand{\weight}[1]{\ensuremath{\textcolor{blue}{#1}}}
\newcommand{\trans}[2]{\ensuremath{\letter{#1}|\weight{#2}}}
\newcommand{\ttrans}[3][]{\xrightarrow[#1]{%
\ifthenelse{\equal{#2}{}}{\weight{#3}}{\trans{#2}{#3}}%
}}

\newcommand{\issuffix}{\preceq_s}

\newcommand{\inA}{A}
\newcommand{\outA}{B}
\newcommand{\inL}{\inA^*}
\newcommand{\outL}{\outA^*}
\newcommand{\outLp}{\outA^+}

\newcommand{\Qb}{\overline{Q}}
\newcommand{\ib}{\overline{i}}

\newcommand{\Contexts}[1]{\ensuremath{\mathcal{C}(#1)}}
\newcommand{\outC}{\Contexts{\outA}}

\newcommand{\out}{\textsf{out}}

\newcommand{\proot}{\rho}
\newcommand{\prootf}{\bar{\rho}}

\usepackage{stmaryrd}

\newcommand*{\TallestContent}{$x$}

\newcommand*{\ra}[1]{\overrightarrow{\makebox{$#1$\vphantom{\TallestContent}}}}
\newcommand*{\la}[1]{\overleftarrow{\makebox{$#1$\vphantom{\TallestContent}}}}

\newcommand{\Lc}[1]{\ensuremath{\la{#1}}}
\newcommand{\Rc}[1]{\ensuremath{\ra{#1}}}

\usepackage[color]{changebar}

\usepackage[many]{tcolorbox}
\newtcolorbox{cross}{blank,breakable,parbox=false,
  overlay={\draw[red,line width=5pt] (interior.south west)--(interior.north east);
    \draw[red,line width=5pt] (interior.north west)--(interior.south east);}}

\newcommand{\Choose}{\ensuremath{\mathsf{choose}}}

\newcommand{\SCom}[1]{\ensuremath{\mathrm{strongly-}#1\mathrm{-commuting}}}
\newcommand{\SAlign}[1]{\ensuremath{\mathrm{strongly-}(#1)\mathrm{-aligned}}}

\newcommand{\pow}{\ensuremath{\mathsf{pow}}}
\renewcommand{\split}{\ensuremath{\mathsf{split}}}
\renewcommand{\extract}{\ensuremath{\mathsf{extract}}}

\newcommand{\Qbi}{\ensuremath{\Qb_\infty}}
\newcommand{\Qbs}{\ensuremath{\Qb_{\textsf{start}}}}
\newcommand{\Qbc}{\ensuremath{\Qb_{\textsf{com}}}}
\newcommand{\Qbnc}{\ensuremath{\Qb_{\neg \textsf{com}}}}

\newcommand{\act}{\bullet}
\newcommand{\ceps}{c_\eps}

\newcommand{\iB}{\overline{i}}
\newcommand{\pB}{\overline{p}}
\newcommand{\qB}{\overline{q}}
\newcommand{\rB}{\overline{r}}
\newcommand{\sB}{\overline{s}}

\newcommand{\case}[1]{\medskip\noindent\underline{#1}\quad}

\newcommand{\const}[1]{\mathsf{#1}}

%% file: figure/macros.tex
\usepackage{xifthen}

\tikzset{
  every node/.style={
    font=\small
  },
  state/.style={
    circle,
    minimum width=0.6cm,
    draw=black!80,
    very thick,
    fill=black!5,
    font=\small
  },
  initial/.style={
    initial by arrow,
    initial where=left,
    initial distance=1cm,
    initial text=,
    append after command={%
      \pgfextra
        \pgfinterruptpath
          \node[above left, xshift=-0.2cm] at (\tikzlastnode.west) {#1};
        \endpgfinterruptpath
      \endpgfextra
    },
  },
  accepting/.style={
    accepting by arrow,
    accepting where=right,
    accepting distance=1cm,
    accepting text=,
    append after command={%
      \pgfextra
        \pgfinterruptpath
          \node[above right, xshift=+0.2cm] at (\tikzlastnode.east) {#1};
        \endpgfinterruptpath
      \endpgfextra
    },
  },
}

\makeatletter

\tikzstyle{accepting by relation}=    [after node path=
{
  {
    [to path=
    {
      [->,double=none,every accepting by relation]
      --
      ([shift=(\tikz@accepting@angle:\tikz@accepting@distance)]\tikztostart.\tikz@accepting@angle)
          node [shape=rectangle,anchor=\tikz@accepting@text@anchor,
            shift=(\tikz@accepting@angle:-\tikz@accepting@distance/2),
          ] {\tikz@accepting@text}
      }]
    edge ()
  }
}]
\tikzstyle{every accepting by relation}=[]

\tikzstyle{initial by relation}=   [after node path=
{
  {
    [to path=
    {
      [->,double=none,every initial by relation]
      ([shift=(\tikz@initial@angle:\tikz@initial@distance)]\tikztostart.\tikz@initial@angle)
          node [shape=rectangle,anchor=\tikz@initial@text@anchor,
            shift=(\tikz@initial@angle:-\tikz@initial@distance/2),
          ] {\tikz@initial@text}
        -- (\tikztostart)}]
    edge ()
  }
}]
\tikzstyle{every initial by relation}=[]

\tikzoption{initial text anchor}{\tikzaddafternodepathoption{\def\tikz@initial@text@anchor{#1}}}
\tikzoption{accepting text anchor}{\tikzaddafternodepathoption{\def\tikz@accepting@text@anchor{#1}}}

\def\tikz@initial@text@anchor{west}
\def\tikz@accepting@text@anchor{west}

\makeatother

%% file: 000-abstract.tex

Transducers extend finite state automata with outputs,
and describe transformations
from strings to strings. 
Sequential transducers,
which have a deterministic behaviour regarding their
input, are of particular interest.
However, unlike finite-state automata, not every transducer can be made
sequential.
The seminal work
of Choffrut allows to characterise, amongst the
functional one-way transducers, the ones
that admit an equivalent sequential transducer.

In this work, we extend the results of Choffrut 
to the class of transducers that produce
their output string by adding simultaneously, at each transition, 
a string on the left and a string on the right of the string 
produced so far. We call them the string-to-context transducers.
We obtain a multiple characterisation
of the functional string-to-context transducers admitting an equivalent sequential one,
based on a Lipschitz property of the function realised
by the transducer, and on a pattern (a new twinning
property). Last, we prove that given a string-to-context transducer,
determining whether there exists an equivalent sequential one
is in \textsf{coNP}.
%
%

%% file: 010-introduction.tex





Transducers are a fundamental model to describe programs manipulating strings.
They date back to the very first works in theoretical computer science, and
are already present in the pioneering works on finite state
automata~\cite{Scott67,AhoHU69}. While finite state automata
are very robust w.r.t. modifications of the model such as non-determinism
and two-wayness, this is not the case for transducers. These two
extensions do affect the expressive power of the model. Non-determinism
is a feature very useful for modelisation and specification purposes.
However, when one turns to implementation,  deriving a sequential,
\emph{i.e.} input-deterministic, transducer is a major issue. A natural
and fundamental problem thus consists, given a (non-deterministic)
transducer, in deciding whether there exists an equivalent sequential
transducer. This problem is called the \emph{sequentiality problem}.







In~\cite{Choffrut77}, Choffrut addressed this problem for the
class of functional (one-way) finite state transducers,
which corresponds to so-called \emph{rational functions}.
He proved a multiple characterisation
of the transducers admitting an equivalent sequential transducer.
This characterisation includes a machine-independent property,
namely a Lipschitz property of the function realised by the transducer.
It also involves a pattern property, namely the twinning property,
that allows to prove that the sequentiality
problem is decidable in polynomial time for the class of
functional finite state transducers~\cite{WK95}.
This seminal work
has led to developments on the sequentiality
of finite state transducers~\cite{BealEtAl03a,BealC02}. These results have
also been extended to 
weighted automata~\cite{BGW00,KM05,Filiot2015}
and to tree transducers~\cite{Seidl93}.
See also~\cite{DBLP:journals/tcs/LombardyS06} for a survey
on sequentiality problems.

%
%

While the model of one-way transducers is now rather well-understood,
a current challenge is to address the so-called class of \emph{regular functions},
which corresponds to functions realised by two-way transducers.
This class has attracted a lot of interest during the last years.
It is closed under composition~\cite{ChytilJ77} and enjoys alternative presentations
using logic~\cite{EngelfrietH01},
a deterministic one-way model equipped with registers, named
streaming string transducers~\cite{AlurCerny10} (SST for short), as well as
a set of regular combinators~\cite{AlurFR14,BaudruR18,DaveGK18}.
This class of functions is much more expressive, as it captures for
instance the mirror image and the copy. Yet, it has good decidability
properties: equivalence and type-checking are decidable in \textsf{PSpace}~\cite{Gurari82,AC11}.
We refer the interested reader to~\cite{FiliotR16} for a recent survey.
Intuitively, two-way finite state transducers
(resp. SST) extend
one-way finite state transducers with two important features: firstly, they can go through the input word
both ways (resp. they can prepend and append words to registers),
and secondly, they can perform multiple passes (resp. they can perform register concatenation).




In this paper, we lift the results of Choffrut~\cite{Choffrut77} to a class of transducers that can
perform the first of the two features mentioned above, thus generalising the class of
rational functions. More precisely, we consider transducers which, at each transition, extend the
output word produced so far by prepending and appending two words to it. This
operation can be defined as the extension of a word with a context,
and we call these transducers the \emph{string-to-context transducers}. However, it is
important to notice that that they still describe functions from strings to strings.
We characterise the functional string-to-context transducers that admit an equivalent
\emph{sequential} string-to-context transducer through $i)$ a machine independent
property: the function realised by the transducer satisfies a Lipschitz property that involves
an original \emph{factor distance} and $ii)$ a pattern property of the transducer which
we call \emph{contextual twinning property}, and that generalises the twinning property to contexts.
We also prove that the sequentiality problem for these transducers
is in the class \textsf{coNP}.

A key technical tool of the result of~\cite{Choffrut77} was a combinatorial analysis
of the loops, showing that the output words of synchronised loops have
conjugate primitive roots. For string-to-context transducers, the situation is more
complex, as the combinatorics may involve the words of the two sides of the context.
Intuitively, when these words do commute with the output word produced so far,
it is possible for instance to move to the right a part of the word produced on the left.
In order to prove our results, we thus
dig into the combinatorics of contexts associated with loops, identifying
different possible situations, and we then use this analysis to describe an original
determinisation construction.


Our results also have a strong connection with the register minimisation
problem for SST. This problem consists in
determining, given an SST and a natural number $k$,
whether there exists an equivalent SST with $k$ registers.
It has been proven in~\cite{DRT16} that the problem is decidable
for SST that can only append words to registers, and the proof
crucially relies on the fact that the $k=1$ case exactly corresponds
to the sequentiality problem of one-way finite state transducers.
Hence, our results constitute a first step towards register minimisation
for SST without register concatenation.
The register minimisation problem for \emph{non-deterministic} SST
has also been studied in~\cite{BGMP16} for the case of concatenation-free SST.
The targeted model being non-deterministic, the two problems are independent.



Due to lack of space, omitted proofs can be found in the Appendix.

%% file: 020-preliminaries.tex

\paragraph*{Words, contexts and partial functions}

Let $\inA$ be a finite alphabet.
  The set of finite words (or strings) over $\inA$ is denoted by $\inA^*$.
  The empty word is denoted by $\epsilon$.
  The length of a word $u$ is denoted by $|u|$.
We say that a word $u$ is a prefix (resp. suffix) of a word $v$
  if there exists a word $y$ such that $uy = v$ (resp. $yu = v$).
We say that two words $u,v\in \inA^*$ are conjugates
if there exist two words $t_1,t_2\in\inA^*$ such that
$u=t_1t_2$ and $v=t_2t_1$. If this holds, we write $u\sim v$.
%
%
The primitive root of a word $u\in \inA^*$, denoted $\rho(u)$, is the shortest word $x$
such that $u=x^p$ for some $p\ge 1$.
\begin{lemma}[\cite{fine_uniqueness_1965}]\label{l:Fine}
Let $u,v\in\inA^*$. There exists $n\in\N$ such that if there is a common
factor of $u$ and $v$ of length at least $n$, then $\rho(u)\sim\rho(v)$.
\end{lemma}

Given two words $u,v\in \inA^*$,
  the \intro{longest common prefix} (resp. \intro{suffix}) of $u$ and $v$
    is denoted by $\lcp(u,v)$ (resp. $\lcs(u,v)$).
  We define the \intro{prefix distance} between $u$ and $v$,
    denoted by $\dist(u,v)$, as $|u|+|v|-2|\lcp(u,v)|$.

Given a word $u\in\outA^*$,
  we say that $v$ is a \intro{factor} of $u$
  if there exist words $x,y$ such that $u=xvy$.
Given two words $u,v\in\outA^*$,
  a \intro{longest common factor} of $u$ and $v$
  is a word $w$ of maximal length that is a factor of both $u$ and $v$.
Note that this word is not necessarily unique. We denote such a word by $\lcf(u,v)$.
The \intro{factor distance} between $u$ and $v$, denoted by $\distf(u,v)$,
  is defined as $\distf(u,v) = |u|+|v| - 2|\lcf(u,v)|$.
This definition is correct
  as $|\lcf(u,v)|$ is independent of the choice of the common factor of maximal length.


Using a careful case analysis, we can prove that $\distf$ is indeed a distance,
  the only difficulty lying in the subadditivity:

\begin{lemma}\label{r:distf}
  $\distf$ is a distance.
\end{lemma}

Given a finite alphabet $\outA$, a \intro{context} on $\outA$ is
  a pair of words $(u,v) \in \outL \times \outL$.
The set of contexts on $\outA$ is denoted $\outC$.
The empty context is denoted by $\ceps$.
For a context $c = (u,v)$,
  we denote by $\Lc{c}$ (resp. $\Rc{c}$)
  its left (resp. right) component: $\Lc{c} = u$ (resp. $\Rc{c} = v$).
  The \intro{length} of a context $c$ is defined by $|c| = |\Lc{c}| + |\Rc{c}|$.
The \intro{lateralized length} of a context $c$
  is defined by $\|c\| = (|\Lc{c}|,|\Rc{c}|)$.
For a context $c \in \outC$ and a word $w \in \outL$,
  we write $c[w]$ for the word $\Lc{c}w\Rc{c}$.
We define the concatenation of
  two contexts $c_1,c_2 \in \outC$
  as the context $c_1 c_2 = (\Lc{c_1}\Lc{c_2}, \Rc{c_2}\Rc{c_1})$.
  Last, given a context $c$ and a word $u$, we denote by
  $c^{-1}[u]$ the unique word $v$ such that $c[v]=u$, when such a word
  exists.

Given a set of contexts $C\subseteq\outC$, we denote
by $\lcc(C)$ the longest common context of elements in $C$,
defined as $\lcc(C) = (\lcs(\{\Lc{c}\mid c\in C\}),\lcp(\{\Rc{c}\mid c\in C\}))$.
We also write $C.\lcc(C)^{-1} = \{c' \mid c'.\lcc(C) \in C\}$.

We consider two sets $X,Y$.
  Given $\Delta \subseteq X \times Y$, we let $\dom(\Delta) =
  \{x\in X \mid \exists y, (x,y)\in \Delta\}$.
We denote the set of partial functions
from $X$ to $Y$ as $\pf{X}{Y}$.
Given $f\in\pf{X}{Y}$, we write $f:X\pto Y$, and we denote by $\dom(f)$ its domain.
When more convenient, we may also see elements of $\pf{X}{Y}$ as subsets of $X\times Y$.
 Last, given $\Delta \subseteq X \times Y$,
  we let $\Choose(\Delta)$
  denote some $\Delta'\in\pf{X}{Y}$ such that
  $\Delta' \subseteq \Delta$ and $\dom(\Delta) = \dom(\Delta')$.

\paragraph*{String-to-Context and String-to-String Transducers}

\begin{definition}
Let $\inA,\outA$ be two finite alphabets.
A \intro{string-to-context transducer} (\StoC{} for short)
  $\T$ from $\inL$ to $\outL$
  is a tuple $(Q, \tinit, \tfinal, T)$
  where $Q$ is a finite set of states,
  $\tinit:  Q \pto \outC$
    (resp. $\tfinal: Q \pto \outC$)
    is the finite initial (resp. final) function,
  $T \subseteq Q \times \inA \times \outC \times Q$ is the finite set of transitions.
\end{definition}

  A state $q$ is said to be \intro{initial} (\resp \intro{final})
  if $q\in\dom(\tinit)$ (resp. $q\in\dom(\tfinal)$).
  We depict as as $\ttrans{}{c} q$ (resp. $q \ttrans{}{c}$)
  the fact that $\tinit(q)=c$ (resp. $\tfinal(q)=c$).
%
%
A run $\rho$ from a state $q_1$ to a state $q_k$ on a word $w = w_1 \dotsm w_k \in \inL$
  where for all $i$, $w_i \in \inA$, is a sequence of transitions:
  $(q_1,w_1,c_1,q_2),(q_2,w_2,c_2,q_3),\ldots,(q_k,w_k,c_k,q_{k+1})$.
  The \intro{output} of such a run is
  the context $c = c_k \dots c_2 c_1 \in \outC$, and is denoted by $\out(\rho)$.
  We depict this situation as $q_1 \ttrans{w}{c} q_{k+1}$.
The set of runs of $\T$ is denoted $\Hc(\T)$.
The run $\rho$ is said to be \intro{accepting} if $q_1$ is initial and $q_{k+1}$ final.
This \longStoC{} $\T$
  computes a relation $\inter{\T} \subseteq \inL \times \outL$
  defined by the set of pairs $(w,edc[\eps])$ such that there are $p,q \in Q$
  with $\ttrans{}{c} p \ttrans{w}{d} q \ttrans{}{e}$.
  Thus, even if its definition involves contexts on $\outA$, the semantics of $\T$
  is a relation between words on $\inA$ and words on $\outA$.
Given an \StoC{}  $\T=(Q, \tinit, \tfinal, T)$,
  we define the constant $\base$ as
  $\base = \max\{|c| \mid (p,a,c,q)\in T
    \text{ or } (q,c) \in \tinit \cup \tfinal\}$.
Given $\Delta: Q \pto \outC$, we denote by $\T_\Delta$ the \StoC{}
obtained by replacing $\tinit$ with $\Delta$.
%
%
An \StoC{}  is \intro{trimmed} if each of its states appears in some accepting run.
  W.l.o.g., we assume that the \longStoC{s} we consider are trimmed.
%
%
An \StoC{}  $\T$ from $\inL$ to $\outL$ is \intro{functional}
  if the relation $\inter{\T}$ is a function from $\inL$ to $\outL$.
%
  An \StoC{}  $\T=(Q, \tinit, \tfinal, T)$ is \intro{sequential}
  if $\dom(\tinit)$ is a singleton
  and if for every transitions $(p,a,c,q), (p,a,c',q')\in T$, we have $q=q'$ and $c=c'$.



The classical model of finite-state transducers is recovered in the following definition:
\begin{definition}
Let $\inA,\outA$ be two finite alphabets.
A \longStoC{} $\T = (Q, \tinit, \tfinal, T)$
  is a \intro{string-to-string transducer} (\NFT{} for short) from $\inL$ to $\outL$
  if,
    for all $(q,c) \in \tinit \cup \tfinal$, $\Lc{c} = \eps$, and
    for all $(q,a,c,q') \in T$, $\Lc{c} = \eps$.
\end{definition}

Notations defined for \StoC{} hold for classical transducers as is.
For an \NFT{}, we write
  $\ttrans{}{w} q$ (resp. $q \ttrans{}{w}$, and $q \ttrans{u}{w} q'$)
  instead of
  $\ttrans{}{(\eps,w)} q$ (resp. $q \ttrans{}{(\eps,w)}$, and $q \ttrans{u}{(\eps,w)} q'$).

Given an \StoC{} $\T=(Q, \tinit, \tfinal, T)$,
  we define its \intro{right \NFT{}}, denoted $\Rc{\T}$,
  as the tuple $(Q, \overrightarrow{\tinit}, \overrightarrow{\tfinal}, \Rc{T})$
    where, for all $q \in Q$,
      $\overrightarrow{\tinit}(q) = \overrightarrow{\tinit(q)}$
      and $\overrightarrow{\tfinal}(q) = \overrightarrow{\tfinal(q)}$,
    and, for all $(p,a,c,q) \in T$, $(p,a,\Rc{c},q) \in \Rc{T}$.
  Its \intro{left \NFT{}} $\Lc{\T}$ is defined similarly,
    and by applying the mirror image on its output labels.

\begin{example} Two examples of \StoC{} (not realisable by \NFT) are depicted on~\Cref{e:StoC}.

  \begin{figure}[htb!]
    \begin{subfigure}[b]{0.5\textwidth}
      \centering
      \scalebox{.9}{\input{figure/example-s2c-mirror}}
      \caption{$T_{mirror}$}
      \label{e:s2c-mirror}
    \end{subfigure}
    \begin{subfigure}[b]{0.5\textwidth}
      \centering
      \scalebox{.9}{\input{figure/example-s2c-partition}}
      \caption{$T_{partition}$}
      \label{e:s2c-partition}
    \end{subfigure}
    \caption{
      \ref{e:s2c-mirror}
        Example of a \StoC{} $T_{mirror}$
        computing the function $f_{mirror} : u_1 \dots u_n \in \{a,b\}^* \mapsto u_n \dots u_1$.
      \ref{e:s2c-partition}
        Example of a \StoC{} $T_{partition}$
        computing the function $f_{partition} : u \in \{a,b\}^* \mapsto a^{|u|_{a}} b^{|u|_{b}}$.
    }\label{e:StoC}
  \end{figure}
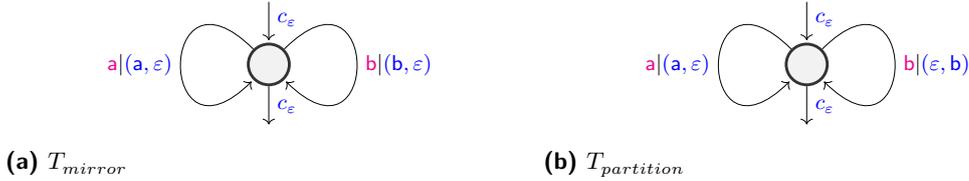
\end{example}

%% file: figure/example-s2c-mirror.tex
\begin{tikzpicture}[
  ->,
  shorten >=1pt,
  node distance=.2cm,
  scale=1,
  initial/.style={
    initial by relation,
    initial where=above,
    initial distance=.6cm,
    initial text=#1,
    initial text anchor=base west,
  },
  accepting/.style={
    accepting by relation,
    accepting where=below,
    accepting distance=.6cm,
    accepting text=#1,
    accepting text anchor=west,
  }
]
  \clip (-2.5,-1) rectangle (2.5, 1);

  \node[initial=\weight{\ceps},accepting=\weight{\ceps},state] (q) at (0,0) {};
  \path (q) edge [loop left, distance=2cm, out=135, in=225]
    node {\trans{\const{a}}{(\const{a},\eps)}} (q);
  \path (q) edge [loop right, distance=2cm, out=45, in=-45]
    node {\trans{\const{b}}{(\const{b},\eps)}} (q);

\end{tikzpicture}

%% file: figure/example-s2c-partition.tex
\begin{tikzpicture}[
  ->,
  shorten >=1pt,
  node distance=.2cm,
  scale=1,
  initial/.style={
    initial by relation,
    initial where=above,
    initial distance=.6cm,
    initial text=#1,
    initial text anchor=base west,
  },
  accepting/.style={
    accepting by relation,
    accepting where=below,
    accepting distance=.6cm,
    accepting text=#1,
    accepting text anchor=west,
  }
]
  \clip (-2.5,-1) rectangle (2.5, 1);

  \node[initial=\weight{\ceps},accepting=\weight{\ceps},state] (q) at (0,0) {};
  \path (q) edge [loop left, distance=2cm, out=135, in=225]
    node {\trans{\const{a}}{(\const{a},\eps)}} (q);
  \path (q) edge [loop right, distance=2cm, out=45, in=-45]
    node {\trans{\const{b}}{(\eps,\const{b})}} (q);

\end{tikzpicture}

%% file: 030-lipschitz-twinning.tex


We recall the properties considered in~\cite{Choffrut77}, and the associated results.

\begin{definition}
We say that a function $f:\inL \pto \outL$ satisfies the \emph{Lipschitz property}
	if there exists $K\in\N$
		such that $\forall u,v \in \dom(f), \dist(f(u),f(v)) \leq K.\dist(u,v)$.
\end{definition}


\begin{definition}
	We consider an \NFT{} and $L \in \N$.
  Two states $q_1$ and $q_2$ are said to be \emph{$L$-twinned}
  if for any two runs
    $\ttrans{}{w_1} p_1 \ttrans{u}{x_1} q_1 \ttrans{v}{y_1} q_1$
  and
    $\ttrans{}{w_2} p_2 \ttrans{u}{x_2} q_2 \ttrans{v}{y_2} q_2$,
  where $p_1$ and $p_2$ are initial,
  we have
  	for all $j \geq 0$, $\dist(w_1 x_1 y_1^j, w_2 x_2 y_2^j) \leq L$.
  An \NFT{} satisfies the \emph{twinning property (TP)}
		if there exists $L \in \N$
			such that any two of its states are $L$-twinned.
\end{definition}

%


\begin{theorem}[\cite{Choffrut77}]
	Let $\T$ be a functional \NFT{}.
	The following assertions are equivalent:
	\begin{enumerate}
		\item there exists an equivalent sequential \NFT{},
		\item $\inter{\T}$ satisfies the Lipschitz property,
		\item $\T$ satisfies the twinning property.
	\end{enumerate}
\end{theorem}


We present the adaptation of these properties to \longStoC{s}.
\begin{definition}
	We say that $f:\inL \pto \outL$ satisfies the \emph{contextual Lipschitz
	property (CLip)}
	if there exists $K\in\N$
		such that $\forall u,v \in \dom(f), \distf(f(u),f(v)) \leq K.\dist(u,v)$.
\end{definition}


\begin{definition}\label{d:ctp}
We consider an \StoC{} and $L \in \N$.
  Two states $q_1$ and $q_2$ are said to be \emph{$L$-contextually twinned}
  if for any two runs
    $\ttrans{}{c_1} p_1 \ttrans{u}{d_1} q_1 \ttrans{v}{e_1} q_1$
  and
    $\ttrans{}{c_2} p_2 \ttrans{u}{d_2} q_2 \ttrans{v}{e_2} q_2$,
     where $p_1$ and $p_2$ are initial,
  %
  we have
		for all $j \geq 0$,
  		$\distf(e_1^j d_1 c_1 [\eps], e_2^j d_2 c_2 [\eps]) \leq L$.
  An \StoC{} satisfies the \emph{contextual twinning property (CTP)}
	if there exists $L \in \N$ such that
		any two of its states are $L$-contextually twinned.
\end{definition}

%% file: 040-main-result.tex

The main result of the paper is the following theorem, which extends
to \longStoC{s} the characterisation of sequential transducers amongst
functional ones.

\begin{theorem}\label{t:main}
  Let $\T$ be a functional \StoC{}.
  The following assertions are equivalent:
  \begin{enumerate}
    \item there exists an equivalent sequential \longStoC,
    \item $\inter{\T}$ satisfies the contextual Lipschitz property,
    \item $\T$ satisfies the contextual twinning property.
  \end{enumerate}
\end{theorem}

\begin{proof}
  The implications $1 \Rightarrow 2$ and $2 \Rightarrow 3$
    are proved in \Cref{r:det-implies-lip}
      and~\Cref{r:lip-implies-ctp} respectively.
  The implication $3 \Rightarrow 1$ is more involved,
    and is based on a careful analysis of word combinatorics
      of loops of \longStoC{s} satisfying the CTP.
  This analysis is summarised in \Cref{r:ctp-implies-2-loop}
    and used in \Cref{sec:construction}
      to describe the construction of an equivalent sequential \StoC{}.
\end{proof}



\begin{proposition}\label{r:det-implies-lip}
  Let $\T$ be a functional \StoC{} realizing the function $f$.
  If there exists an equivalent sequential \StoC{},
    then $f$ satisfies the contextual Lipschitz property.
\end{proposition}

\begin{proof}
  Let us consider $\T'$ the equivalent sequential \StoC{}.
  We claim that $f$ is context-Lipschitzian with coefficient $3\basep$.
  Consider two input words $u,v$ in the domain of $f$.
  If $u=v$, then the result is trivial.
  Otherwise, let $w=\lcp(u,v)$ and let $u=w.u'$, with $0 \le |u'|$.
  Then we have
    $
      \inter{\T'}(u) = c_3 c_2 c_1 [\epsilon]
    $
    where $c_1$ is the context produced along $w$,
      $c_2$ the one produced along $u'$,
      and $c_3$ is the final output context.
  Similarly, we can write (with $v=w.v'$, and $0 \le |v'|$)
    $
      \inter{\T'}(v)=d_3 d_2 d_1 [\epsilon]
    $.
  As $\T'$ is sequential, we have $d_1 = c_1$.
  We also have $|c_3|\leq \basep$,
    $|d_3|\leq \basep$,
    $|c_2|\leq \basep.|u'|$
    and $|d_2|\leq \basep.|v'|$.
  Finally, as $u \neq v$, we have $\dist(u,v) = |u'|+|v'| \geq 1$
  and we obtain:
    $$
    \distf(f(u),f(v)) \leq |c_3 c_2| + |d_3 d_2|
      \leq \basep.(2+|u'|+|v'|)
      \le 3\basep \dist(u,v)
     \qquad \qedhere
    $$
\end{proof}

\begin{proposition}\label{r:lip-implies-ctp}
  Let $\T$ be a functional \StoC{} realizing the function $f$.
  If $f$ satisfies the contextual Lipschitz property,
    then $\T$ satisfies the contextual twinning property.
\end{proposition}

\begin{proof}
  We consider an instance of the CTP and stick to the notations of~\Cref{d:ctp}.
  We denote by $n$ the number of states of $\T$.
  As $\T$ is trimmed, there exist runs $q_i \ttrans{w_i}{f_i} r_i \ttrans{}{g_i}$,
  with $|w_i|\le n$, for $i\in\{1,2\}$.
  We consider the input words $\alpha_j = u v^j w_1$ and $\beta_j = u v^j w_2$,
  for all $j\ge 0$.
  We have, for every $j$, $\dist(\alpha_j,\beta_j) \leq |w_1|+|w_2| \le 2 n$.



The following property of $\distf$ can be proven using a case analysis:\\
\noindent\textbf{Fact.} For every $w,w'\in\outA^*$, $c,c'\in\outC$,
we have $\distf(w,w')\le\distf(c[w],c'[w'])+|c|+|c'|$.

  As $f$ is $K$ context-Lipschitzian, for some fixed $K$,
  we obtain, for all $j$:
$$ \begin{array}{lll}
\distf(e_1^j d_1 c_1 [\eps], e_2^j d_2 c_2 [\eps])
&\le& \distf(g_1 f_1 e_1^j d_1 c_1 [\eps], g_2 f_2 e_2^j d_2 c_2 [\eps]) + 2(n+1)\base\\
&\le& \distf(f(\alpha_j), f(\beta_j))+2(n+1)\base\\
&\le& K\dist(\alpha_j,\beta_j) + 2(n+1)\base \le 2Kn+2(n+1)\base \quad \qedhere
  \end{array}
$$

\end{proof}

%% file: 050-combinatorics.tex


The classical twinning property
  forces the outputs of two runs reading the same input to only diverge by a finite amount.
This constraint in turn makes for strong combinatorial bindings between runs involving loops:
for two runs
  $\ttrans{}{w_1} p_1 \ttrans{u}{x_1} q_1 \ttrans{v}{y_1} q_1$
and
  $\ttrans{}{w_2} p_2 \ttrans{u}{x_2} q_2 \ttrans{v}{y_2} q_2$,
we have $|y_1| = |y_2|$, and $\rho(y_1) \sim \rho(y_2)$.
%
%
Similar behaviours are expected with \longStoC{s} and lead us to study the combinatorial properties of synchronised runs involving loops in those machines.
Throughout this section,
  we consider a \longStoC{} $\T=(Q, \tinit, \tfinal, T)$
  that satisfies the contextual twinning property.

\subsection{Behaviours of Loops}

We start with two examples illustrating how output contexts of synchronised loops
can be modified to obtain an equivalent sequential \StoC{}.


\begin{example}\label{e:commuting-aligned}
  \Cref{e:s2c-com-non-det} shows an example of
    a non-sequential functional \StoC{} transducer $\T_1$.
    The contexts produced on loops around
    states $q_1$ and $q_2$ both commute with word $a$.
This observation can be used to build
  an equivalent sequential \StoC{} $\D_1$, depicted
  on~\Cref{e:s2c-com-det}.
    \Cref{e:s2c-align-non-det} shows an example of
    a non-sequential functional \StoC{} transducer $\T_2$
    where output contexts are non-commuting, but can be slightly shifted so as to be aligned.
This observation can be used to build
 an equivalent sequential \StoC{} $\D_2$, depicted on~\Cref{e:s2c-align-det}.

  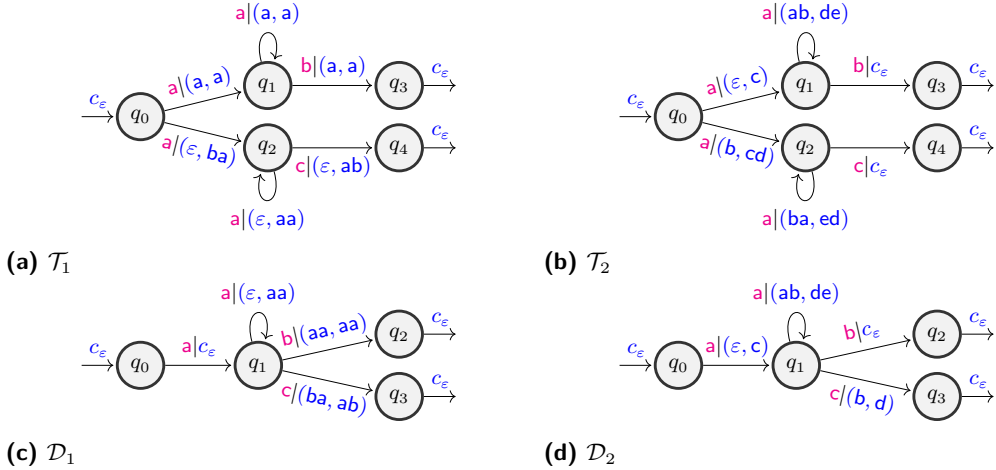
\begin{figure}[htb!]
    \begin{subfigure}[b]{0.5\textwidth}
      \centering

      \def\titop{\trans{\const{a}}{(\const{a},\const{a})}}
      \def\tltop{\trans{\const{a}}{(\const{a},\const{a})}}
      \def\tetop{\trans{\const{b}}{(\const{a},\const{a})}}

      \def\tibot{\trans{\const{a}}{(\eps,\const{ba})}}
      \def\tlbot{\trans{\const{a}}{(\eps,\const{aa})}}
      \def\tebot{\trans{\const{c}}{(\eps,\const{ab})}}

      \scalebox{.9}{\input{figure/example-s2c-behavior-non-det}}
      \caption{$\T_1$}
      \label{e:s2c-com-non-det}
    \end{subfigure}
    \begin{subfigure}[b]{0.5\textwidth}
      \centering

      \def\titop{\trans{\const{a}}{(\eps,\const{c})}}
      \def\tltop{\trans{\const{a}}{(\const{ab},\const{de})}}
      \def\tetop{\trans{\const{b}}{\ceps}}

      \def\tibot{\trans{\const{a}}{(\const{b},\const{cd})}}
      \def\tlbot{\trans{\const{a}}{(\const{ba},\const{ed})}}
      \def\tebot{\trans{\const{c}}{\ceps}}

      \scalebox{.9}{\input{figure/example-s2c-behavior-non-det}}
      \caption{$\T_2$}
      \label{e:s2c-align-non-det}
    \end{subfigure}

    \begin{subfigure}[b]{0.5\textwidth}
      \centering

      \def\ti{\trans{\const{a}}{\ceps}}
      \def\tl{\trans{\const{a}}{(\eps,\const{aa})}}
      \def\tetop{\trans{\const{b}}{(\const{aa},\const{aa})}}
      \def\tebot{\trans{\const{c}}{(\const{ba},\const{ab})}}

      \scalebox{.9}{\input{figure/example-s2c-behavior-det}}
      \caption{$\D_1$}
      \label{e:s2c-com-det}
    \end{subfigure}
    \begin{subfigure}[b]{0.5\textwidth}
      \centering

      \def\ti{\trans{\const{a}}{(\eps,\const{c})}}
      \def\tl{\trans{\const{a}}{(\const{ab},\const{de})}}
      \def\tetop{\trans{\const{b}}{\ceps}}
      \def\tebot{\trans{\const{c}}{(\const{b},\const{d})}}

      \scalebox{.9}{\input{figure/example-s2c-behavior-det}}
      \caption{$\D_2$}
      \label{e:s2c-align-det}
    \end{subfigure}

    \caption{
      \ref{e:s2c-com-non-det}
        An \StoC{} $\T_1$
          computing the function that maps
            $a^n b$ to $a^{2n+2}$ and
            $a^n c$ to $b a^{2n} b$.
      \ref{e:s2c-com-det}
        A sequential \StoC{} $\D_1$ equivalent to $\T_1$.
      \ref{e:s2c-align-non-det}
        An \StoC{} $\T_2$
          computing the function that maps
            $a^n b$ to $(ab)^{n-1} c (de)^{n-1}$ and
            $a^n c$ to $b (ab)^{n-1} c (de)^{n-1} d$.
      \ref{e:s2c-align-det}
        A sequential \StoC{} $\D_2$ equivalent to $\T_2$.
    }
    \label{e:s2c-loops}
  \end{figure}
\end{example}

The following definition follows from the intuition drawn by the previous example.




\newcommand{\Com}[1]{\ensuremath{#1\mathrm{-commuting}}}
\newcommand{\Align}[1]{\ensuremath{(#1)\mathrm{-aligned}}}


\begin{definition}[Lasso, Aligned/Commuting/Non-commuting lasso]
  \label{d:commuting-lasso}\label{d:aligned-lasso}
  A \intro{lasso} around a state $q$ is a run $\rho$ of the form
    $\ttrans{}{c} p \ttrans{u}{d} q \ttrans{v}{e} q$
    with $p$ an initial state.
  $\rho$ is said to be \intro{productive}, if $|e| \neq 0$.
  We say that $\rho$ is:
\begin{itemize}
  \item \intro{aligned} w.r.t. $f$ and $w$,
    for some $f \in \outC$ and $w \in \outL$,
  denoted as \intro{$\Align{f,w}$},
    if there exists a context $g \in \outC$
    such that
      for all $i \in \N$,
        $e^i d c [\eps] = g f^i [w]$.
  \item \intro{commuting} w.r.t. $x$,
    for some $x \in \outA^+$,
  denoted as \intro{$\Com{x}$},
    if there exists a context $f \in \outC$
    such that
      for all $i \in \Nplus$, there exists $k \in \N$ such that
        $e^i d c [\eps] = f [x^k]$.
  \item \intro{non-commuting}
  if there exists no word $x \in \outA^+$
  such that $\rho$ is commuting w.r.t $x$.
\end{itemize}
  Two lassos
    $\ttrans{}{c_1} p_1 \ttrans{u_1}{d_1} q_1 \ttrans{v_1}{e_1} q_1$ and
    $\ttrans{}{c_2} p_2 \ttrans{u_2}{d_2} q_2 \ttrans{v_2}{e_2} q_2$
  are said to be \intro{synchronised} if $u_1 = u_2$ and $v_1 = v_2$.
  They are said to be \intro{strongly balanced} if $\|e_1\| = \|e_2\|$.
\end{definition}


Given an integer $k\ge 1$, we consider the $k$-th power of $\T$, that we denote by $\T^k$.
A run in $\T^k$ naturally corresponds to $k$ synchronised runs in $\T$, \emph{i.e.}
on the same input word.
We lift the notion of lasso to $\T^k$, and we denote them by $H_1H_2$,
  where $H_1$ starts in initial states
    and ends in some state $v=(q_i)_{i\in\{1,\ldots,k\}}\in Q^k$,
  and $H_2$ is a loop around state $v$.
In the sequel, we will only consider lassos
  such that $v$ contains pairwise distinct states ($q_i\neq q_j$ for all $i\neq j$).
Those lassos are included in the lassos in $\T^{\leq |Q|} = \cup_{1 \leq k \leq |Q|} \T^k$.

The intuition given by  \Cref{e:commuting-aligned} is formalised in the following Lemma:


\begin{lemma}\label{r:all-commuting-or-aligned}
  \label{d:split-c}\label{d:split-nc}
  Let $H_1 H_2 =  (\rho_j)_{j \in \{1,\dots,k\}}$ a lasso in $\T^k$, for some $1 \leq k \leq |Q|$.
  We write $\rho_j :\ \ttrans{}{c_j} p_j \ttrans{u_1}{d_j} q_j \ttrans{u_2}{e_j} q_j$
  for each $j$.
  Then there exists an integer $m\in \N$ such that
   $|e_j|=m$ for all $j\in\{1,\ldots,k\}$.
  If $m>0$, we say that the lasso $H_1H_2$ is productive, and:
  \begin{itemize}
    \item either there exists $x \in \outA^+$ primitive
        such that $\rho_j$ is $\Com{x}$ for all $j\in \{1,\dots,k\}$. In this case,
        we say that the lasso $H_1H_2$ is $\Com{x}$, and
	we let $\pow_c(x, H_1, H_2) = m/{|x|}$ and
	$\split_c(x, H_1, H_2) = \{(q_j, f_j) \mid j \in \{1,\dots,k\}\}$ where
	$f_j\in \outC$ is such that
	$\forall \alpha \in \N, e_j^\alpha d_j c_j [\eps] = f_j [x^{\alpha\ \pow_c(x, H_1, H_2)}]$.
    \item or there exist $f \in \outC$ and $w \in \outL$ such that
    $\rho_j$ is non-commuting and $\Align{f,w}$ for all $j\in \{1,\dots,k\}$. In this case,
        we say that the lasso $H_1H_2$ is $\Align{f,w}$, and
	we let $\split_{nc}(f, w, H_1, H_2) = \{
    (q_j, g_j) \mid j \in \{1,\dots,k\} \}$ where $g_j\in \outC$ is such that
        $\forall \alpha \in \N, e_j^\alpha d_j c_j [\eps] = g_j f^\alpha [w]$.
  \end{itemize}
\end{lemma}


\begin{proof}[Proof Sketch]
   As $\T$ satisfies the CTP, the outputs must grow at the same pace when
   the loops are pumped. This  entails that the lengths of the $e_j$ must be equal.
  %
  %
  Next, the result is proved by considering two productive synchronised lassos, with
  loops producing respectively $e_1$ and $e_2$.
  If they are not strongly balanced or one of them is $\Com{x}$, for some $x \in \outLp$,
    then, using the result of Fine and Wilf (\Cref{l:Fine}) between $\Lc{e_1},\Lc{e_2},\Rc{e_1}$
    and $\Rc{e_2}$, we can prove that
    the other one is also $\Com{x}$.
  Otherwise, they are both non-commuting and strongly balanced. Using again \Cref{l:Fine}
  but first between $\Lc{e_1}$ and $\Lc{e_2}$, and then
  between $\Rc{e_1}$ and $\Rc{e_2}$, we prove that there exist $f\in\outC$ and
  $w\in\outL$ such that $\rho_1$ and $\rho_2$ are $\Align{f,w}$.
  Finally, the result is lifted to $k$ productive synchronised lassos.
\end{proof}


\begin{example}
  We consider the example \StoC{} in \Cref{e:s2c-loops}.
  The lasso in $\T_1^2$ around $(q_1,q_2)$ is $\Com{\const{a}}$.
  We can compute a $\pow_c$ of $2$
    and $\{(q_1,(\const{a},\const{a})),(q_2,(\const{b},\const{a}))\}$ as a possible $\split_c$.
  The lasso in $\T_2^2$ around $(q_1,q_2)$ is $\Align{(\const{ab},\const{de}),\const{c}}$.
  We can compute $\{(q_1,\ceps),(q_2,(\const{b},\const{d}))\}$ as a possible $\split_{nc}$.
\end{example}

\subsection{Analysis of Loops Consecutive to a Productive Loop}

Consider a run that contains two consecutive productive loops.
We can observe that the type (commuting or non-commuting) of the lasso involving the first loop
  impacts the possible types of the lasso involving the second loop.
For instance, it is intuitive that a non-commuting lasso cannot be followed by a commuting lasso.
Similarly, an $\Com{x}$ lasso cannot be followed by an $\Com{y}$ lasso,
  if $x$ and $y$ are not conjugates.
We will see that loops following a first productive loop
  indeed satisfy stronger combinatorial properties.
The following definition characterises their properties.


\begin{definition}[Strongly commuting/Strongly aligned lasso]
  \label{d:strongly-commuting-lasso}\label{d:strongly-aligned-lasso}
  Let $\rho$ be a productive lasso $\ttrans{}{c} p \ttrans{u}{d} q \ttrans{v}{e} q$
    and $x \in \outA^+$.
  We say that $\rho$ is:
  \begin{itemize}
    \item \intro{strongly commuting} w.r.t. $x$,
      denoted as \intro{$\SCom{x}$},
        if there exists a context $f \in \outC$
        such that
          for all $i,j \in \Nplus$, there exists $k \in \N$ such that
            $e^i d c [x^j] = f [x^k]$.
    \item \intro{strongly aligned} w.r.t. $g$, $f$ and $x$,
    denoted as \intro{$\SAlign{g,f,x}$},
      if there exists a context $h \in \outC$
      such that
        for all $i,j \in \N$,
          $e^j d c [x^i] = h g^j f [x^i]$.
  \end{itemize}
  %
  %
\end{definition}

The following Lemma states the properties of a lasso consecutive to a commuting lasso.
To prove it, we proceed as for \Cref{r:all-commuting-or-aligned} by proving the
result first for two runs and then lifting it to $k$ runs.
The case of two runs is obtained by distinguishing whether they are
strongly balanced or not, and using~\Cref{l:Fine}.

\begin{lemma}\label{r:all-strongly-commuting-or-aligned}
  Let $H_1 H_2$ a productive $\Com{x}$ lasso in $\T^{\leq |Q|}$,
      for some $x \in \outLp$.
  Let $\Delta = \split_{c}(x, H_1, H_2)$
    and $H_3 H_4 = (\rho_j)_{j \in \{1,\dots,k\}}$
      a productive lasso in $\T^k_{\Delta}$, for some $1 \leq k \leq |Q|$.
      We write $\rho_j:\
      \ttrans{}{c_j} p_j \ttrans{u_1}{d_j} q_j \ttrans{u_2}{e_j} q_j$ for each $j$.
  Then:
  \begin{itemize}
    \item either every $\rho_j$ is $\SCom{x}$: we say that $H_3 H_4$ is $\SCom{x}$,
    \item or
      there exist $g,h \in \outC$ such that every $\rho_j$ is $\SAlign{h,g,x}$. In this case, we say that
      $H_3 H_4$ is $\SAlign{h,g,x}$
      and
      we let $\extract_{nc}(h, g, x, \Delta, H_3, H_4) =
        \{ (q_j, h_j) \mid j \in \{1,\dots,k\} \}$
      where $h_j \in \outC$ is s.t.
        $\forall \alpha,\beta \in \N, e_j^\alpha d_j c_j [x^\beta]= h_j h^\alpha g [x^\beta]$.
  \end{itemize}
\end{lemma}



The following Lemma states that once a non-commuting loop is encountered,
then the alignment of production is fixed, \emph{i.e.} no transfer
between left and right productions is possible anymore. Hence, the left
and right  \NFT{}  derived from the \StoC{} both satisfy the twinning property:
\begin{lemma}\label{r:all-fully-aligned}
  Let $H_1 H_2$ be a productive non-commuting lasso
    that is either
  \begin{itemize}
    \item $\Align{f,w}$ in $\T^{\leq |Q|}$,
      for some $f \in \outC$ and $w \in \outL$,
      and $\Delta' = \split_{nc}(f, w, H_1, H_2)$,
    \item or $\SAlign{g,f,x}$ in $\T^{\leq |Q|}_\Delta$,
      for some $g,f \in \outC$ and $\Delta \in \pf{Q}{\outC}$,
      and $\Delta' = \extract_{nc}(g,f,x,\Delta, H_1, H_2)$.
  \end{itemize}
  Then $\Lc{\T_{\Delta'}}$ and $\Rc{\T_{\Delta'}}$ both satisfy the twinning property.
\end{lemma}


\subsection{A Two-loop Pattern Property}

\newcommand{\twoloop}{2-loop property}

The following \twoloop{} summarises
  the combinatorial properties of the synchronised runs involving loops
  in \longStoC{s} that satisfy the CTP.


\begin{definition}[\twoloop]\label{2-loop}
  Given four runs $H_1,H_2,H_3,H_4$ in $\T^{\leq |Q|}$, such that
  $H_1H_2$ and $(H_1H_3)H_4$ are lassos in $\T^{\leq |Q|}$,
    we say that they satisfy the \twoloop{} if:
  \begin{enumerate}
  \item $H_1 H_2$ is
    either non productive,
    or productive and $\Com{x}$, for some $x \in \outLp$,
    or productive, non-commuting and $\Align{f,w}$, for some $f \in \outC$ and $w \in \outL$.
  \item if $H_1H_2$ is productive and $\Com{x}$,
    we let $\Delta = \split_c(x, H_1, H_2)$,
    then $H_3H_4$ is a lasso in $\T_\Delta^{\leq |Q|}$.
    If productive then it is:
    \begin{enumerate}
      \item either $\SCom{x}$,
      \item or non-commuting and $\SAlign{h,g,x}$, for some $g,h \in \outC$.
          We let $\Delta' = \extract_{nc}(h, g, x, \Delta, H_3, H_4)$,
          then $\Lc{\T_{\Delta'}}$ and $\Rc{\T_{\Delta'}}$ both satisfy the twinning property.
    \end{enumerate}
  \item if $H_1H_2$ is productive, non-commuting and $\Align{f,w}$,
    we let $\Delta = \split_{nc}(f, w , H_1, H_2)$,
    then $\Lc{\T_{\Delta}}$ and $\Rc{\T_{\Delta}}$ both satisfy the twinning property.
  \end{enumerate}

  A \longStoC{} $\T$ is said to satisfy the \twoloop{}
  if for all runs $H_1,H_2,H_3,H_4$ as above, they satisfy the \twoloop{}.
%
\end{definition}

As a consequence of~\Cref{%
      r:all-commuting-or-aligned,%
      r:all-strongly-commuting-or-aligned,%
      r:all-fully-aligned%
    }, we have:
\begin{lemma}\label{r:ctp-implies-2-loop}
  If an \StoC{} $\T$ satisfies the CTP
  then it satisfies the \twoloop.
\end{lemma}

%% file: figure/example-s2c-behavior-non-det.tex
\begin{tikzpicture}[
  ->,
  shorten >=1pt,
  node distance=1.2cm,
  scale=1,
  initial/.style={
    initial by relation,
    initial distance=.5cm,
    initial text=#1,
    initial text anchor=south,
  },
  accepting/.style={
    accepting by relation,
    accepting distance=.5cm,
    accepting text=#1,
    accepting text anchor=south,
  },
]
  \clip (-1,-1.7) rectangle (4.75, 1.7);

  \node[initial=\weight{\ceps},state] (qi) {$q_0$};
  \node[state] (q1) [above right=of qi,yshift=-.9cm,xshift=.5cm] {$q_1$};
  \node[state] (q2) [below right=of qi,yshift=.9cm,xshift=.5cm] {$q_2$};
  \node[accepting=\weight{\ceps},state] (qf1) [right=of q1] {$q_3$};
  \node[accepting=\weight{\ceps},state] (qf2) [right=of q2] {$q_4$};

  \path (qi) edge node [above, sloped] {\titop} (q1);
  \path (q1) edge [loop above] node    {\tltop} (q1);
  \path (q1) edge node [above, sloped] {\tetop} (qf1);

  \path (qi) edge node [below, sloped] {\tibot} (q2);
  \path (q2) edge [loop below] node    {\tlbot} (q2);
  \path (q2) edge node [below, sloped] {\tebot} (qf2);

\end{tikzpicture}

%% file: figure/example-s2c-behavior-det.tex
\begin{tikzpicture}[
  ->,
  shorten >=1pt,
  node distance=1.2cm,
  scale=1,
  initial/.style={
    initial by relation,
    initial distance=.5cm,
    initial text=#1,
    initial text anchor=south,
  },
  accepting/.style={
    accepting by relation,
    accepting distance=.5cm,
    accepting text=#1,
    accepting text anchor=south,
  },
]
  \clip (-1, -0.85) rectangle (4.75, 1.25);

  \node[initial=\weight{\ceps},state] (qi) {$q_0$};
  \node[state] (q1) [right=1cm of qi] {$q_1$};
  \node[accepting=\weight{\ceps},state] (qf1) [above right=of q1,yshift=-.9cm,xshift=.7cm] {$q_2$};
  \node[accepting=\weight{\ceps},state] (qf2) [below right=of q1,yshift=.9cm,xshift=.7cm] {$q_3$};

  \path (qi) edge node [above, sloped] {\ti} (q1);
  \path (q1) edge [loop above] node    {\tl} (q1);
  \path (q1) edge node [above, sloped] {\tetop} (qf1);
  \path (q1) edge node [below, sloped] {\tebot} (qf2);

\end{tikzpicture}

%% file: 060-construction.tex

Throughout this section,
  we consider a \longStoC{} $\T=(Q, \tinit, \tfinal, T)$
  from $\inA^*$ to $\outA^*$
  that satisfies the \twoloop. Intuitively, our construction
  stores the set of possible runs of $\T$, starting in an initial state,
  on the input word read so far.
  These runs are incrementally simplified by
  erasing synchronised loops, and by replacing a prefix by a partial
  function $\Delta:Q\pto\outC$. These simplifications
  are based on the \twoloop.


%

\subparagraph*{Observation} It is worth noticing that,  as $\T$ is functional,
  if two runs reach the same state, it is
  safe to keep only one of them. This allows us
  to maintain a set of at most $|Q|$ runs.
%

%
\subparagraph*{Notations}
Given $\Delta \in\pf{Q}{\outC}$, $c \in \outC$, $w \in \outL$, $a\in \inA$
and $H \in \Hc(\T^{\leq |Q|})$, we define the following notations and operations:
\begin{itemize}
\item $\Delta c = \{ (q,dc) \mid (q,d) \in \Delta \}$,
\item $\Delta [w] = \{ (q,d[w]) \mid (q,d) \in \Delta \}$,
\item $\Delta \act a = \Choose(\{ (q', dc) \mid (q,c) \in \Delta \text{ and } q \ttrans{a}{d} q' \})$,
\item $H \act a  \in \Hc(\T^{\leq |Q|})$ is the run obtained by
    extending runs of $H$ with consecutive transitions of $\T$
    associated with input symbol $a$, and by eliminating runs so as to ensure
    that runs reach pairwise distinct states of $\T$,
\item $\Delta \act H = \Choose(\{ (q', dc) \mid (q,c) \in \Delta \text{ and there is a run }
\rho:q\ttrans{x}{d}q' \in H\})$,
\item $id_\Delta = (q_i)_{1\le i \le k}\in\Hc(\T^k)$, for some enumeration $\{q_1,\ldots,q_k\}$
 of $\dom(\Delta)$.
\end{itemize}
%
%


\subparagraph*{Construction}
We define an equivalent deterministic \longStoC{}
  $\overline{\D}=(\Qb, \overline{\tinit}, \overline{\tfinal}, \overline{T})$, and we denote by $\D$
  its trim part. While $\overline{\D}$ may have infinitely many states, we will prove that
  $\D$ is finite. Formally, we define $\Qb = \Qbs \uplus \Qbc \uplus  \Qbnc$
  where:
  \begin{itemize}
    \item $\Qbs = \{ (\eps, \tinit, H) \mid H \in \Hc(\T^{\leq |Q|}) \}$
    \item $\Qbc = \{ (x, \Delta, H) \mid
      x \in \outLp, \Delta \in \pf{Q}{\outC}, H \in \Hc(\T^{\leq |Q|}) \}$
    \item $\Qbnc = \{ (\bot, \Delta, id_\Delta) \mid \Delta \in \pf{Q}{\outC} \}$.
  \end{itemize}

By definition, we have
$\Qb \subseteq (\outL \cup \{\bot\}) \times \pf{Q}{\outC} \times \Hc(\T^{\leq |Q|}) = \Qbi$.
Given $\qB = (x, \Delta, H) \in \Qbi$, we let
   $\Delta_{\qB} = \Delta \bullet H \in \pf{Q}{\Cont{}}$. An invariant of our
  construction is that every starting state of a run in $H$ belongs to $\dom(\Delta)$.

Intuitively, the semantics of a state $\qB = (x, \Delta, H) \in \Qb$ can be understood as follows:
$x$ is used to code the type of state ($\Qbs$, $\Qbc$ or $\Qbnc$), and $\Delta$ and $H$ are used
to represent the runs that remain to be executed to faithfully simulate
the runs of $\T$ on the input word $u$ read so far.
%
%
As we have seen in the previous section, loops may either be commuting, allowing
to shift some parts of the output from one side of the context to the other side, or
they are non-commuting, and then should be aligned, forbidding such modifications.
Intuitively, states in $\Qbs$ correspond
to situations in which no productive loop has been encountered yet. States
in $\Qbc$ (with $x \in \outA^+$) correspond to situations in which only
$x$-commuting loops have been encountered.
States in $\Qbnc$ correspond to situations in which a non-commuting loop has been
encountered. A representation of $\D$ is given in~\Cref{f:states}.

\begin{figure}[htb!]
  \centering
  \scalebox{.9}{\input{figure/sketch-states}}
  \caption{A schematic representation of states and transitions of $\D$.}
  \label{f:states}
\end{figure}
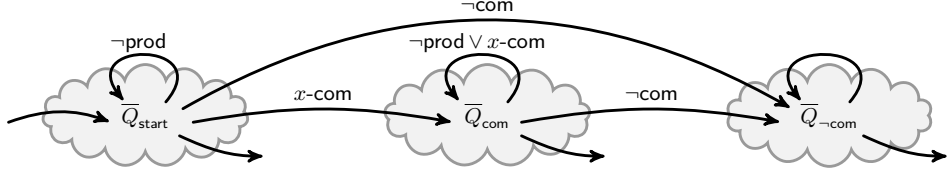




\subparagraph*{Initial and final states} They are defined as follows:
\begin{itemize}
\item $\overline{\tinit} = \{ (\ib, \ceps) \}$ where $ \ib = (\eps, \tinit, id_{\tinit} ) \in \Qbs$
\item $\overline{\tfinal} = \Choose(\{(\bar{q},dc)\mid \qB\in \Qb, (p,c)\in \Delta_{\bar{q}}, (p,d) \in \tfinal \})$ 
\end{itemize}



%

\newcommand{\Let}{\textbf{let} \;}

\begin{algorithm}
  \caption{Extending a state $\pB=(x,\Delta,H_1)\in \Qbs \cup \Qbc$ with $H_2\in \Hc(\T^{\le |Q|})$
  s.t. $H_1H_2$ is a lasso in $\T^{\le |Q|}_\Delta$.}
  \label{extend-with-loop}
  \begin{algorithmic}[1]
    \Function{extend\_with\_loop}{$\pB, H_2$}
      \If{$H_2$ is non-productive} \label{extend-np}
        \State \Return $(\pB, c_\eps)$
      \ElsIf{$\pB = (\eps, \tinit, H_1)$} \label{extend-s}
        \If{$H_1 H_2$ is $\Com{x}$, for some $x \in \outLp$,}
          \State \Let $\Delta = \split_c(x, H_1, H_2)$
          and $k = \pow_c(x, H_1, H_2)$
          \State \Return $((x, \Delta, id_{\Delta}), (\eps, x^k))$
        \ElsIf{$H_1 H_2$ is $\Align{f,w}$, for some $f \in \outC$ and $w \in \outL$,}
          \State \Let $\Delta = \split_{nc}(f, w, H_1, H_2)$
          \State \Return $((\bot, \Delta, id_{\Delta}), f \cdot (\eps, w))$
        \EndIf
      \ElsIf{$\pB = (x, \Delta_0, H_1)$, where $x \in \outLp$,} \label{extend-c}
        \If{$H_1 H_2$ is $\SCom{x}$}
          \State \Let $k = |\out(H_2)|/|x|$
          \State \Return $(\pB, (\eps, x^k))$
        \ElsIf{$H_1 H_2$ is $\SAlign{g,f,x}$, for some $g,f \in \outC$,}
          \State \Let $\Delta = \extract_{nc}(g, f, x, \Delta_0, H_1, H_2)$
          \State \Return $((\bot, \Delta, id_{\Delta}), gf)$
        \EndIf
      \EndIf
    \EndFunction
    \algstore{bkbreak}
  \end{algorithmic}
\end{algorithm}

\vspace{-.2cm}

\subparagraph*{Transitions} Intuitively, a transition of $\overline{\D}$
leaving some state $\pB = (x, \Delta, H) \in \Qb$ with letter $a\in\inA$ aims at
first extending $H$ with $a$, obtaining the new set of runs $H\act a$,
and then simplifying this set of runs by removing loops.
Formally,
we let $(\qB, c) = \Call{simplify}{(x, \Delta, H \act a)}$
  and define the transition $\pB \ttrans{a}{c} \qB$.
 The function $\Call{simplify}{}$
is performed by Algorithm~\ref{simplify}, which calls Algorithm~\ref{extend-with-loop}
to remove all loops of $H \act a$ one by one. Depending on the type
of the loop encountered, the type of the state is updated.
These two algorithms are described below.


%
We first define $\Call{extend\_with\_loop}{\pB, H_2}$ in Algorithm~\ref{extend-with-loop}
  that takes as input a state $\pB=(x,\Delta,H_1)\in \Qbs \cup \Qbc$ and a run $H_2$
  in $\T^{\le |Q|}$ such that $H_1H_2$
  is a lasso in $\T^{\le |Q|}_\Delta$. The algorithm enumerates
  the possible cases for the type of this lasso, depending on the type of $\pB$. This
  enumeration strongly relies on the \twoloop. Depending on the case, the
  loop is processed, and a pair composed of a new state and
  a context is returned.
  This context will be part of the output associated with the transition.
  By a case analysis, we prove:

\begin{lemma}
\label{r:corr-extend}
  Let $\pB = (x,\Delta,H_1) \in \Qbs \cup \Qbc$
    and $H_2 \in \Hc(\T^{\leq |Q|})$
    such that $H_1H_2$
  is a lasso in $\T^{\le |Q|}_\Delta$. We let $(\qB, c) = \Call{extend\_with\_loop}{\pB, H_2}$.
  \begin{itemize}
    \item If $x = \eps$
      then $(\Delta_{\pB} \act H_2) [\eps] = \Delta_{\qB} c [\eps]$.
    \item If $x \in \outLp$
      then for all $k \in \N$, $(\Delta_{\pB} \act H_2) [x^k] = \Delta_{\qB} c [x^k]$.
  \end{itemize}
\end{lemma}

We then define $\Call{simplify}{\pB}$ in Algorithm~\ref{simplify} that takes as
input a state $\pB \in \Qbi$ (we need to consider $\Qbi$ as input and not only $\Qb$ because of
the recursive calls) and returns a pair composed of a new state and a
context. Intuitively, it recursively processes the lassos present in the runs
stored by the state $\pB$, by using calls to the previous algorithm.
The following result is proved by induction, using~\Cref{r:corr-extend}:
%
%
\begin{lemma}
\label{r:corr-simplify}
  Let $\pB = (x,\Delta,H) \in \Qbi$ and $(\qB, c) = \Call{simplify}{\pB}$. Then $\qB \in \Qb$ and we have:
  \begin{itemize}
    \item If $x = \eps$
      then $\Delta_{\pB} [\eps] = \Delta_{\qB} c [\eps]$.
    \item If $x \in \outLp$
      then for all $k \in \N$, $\Delta_{\pB} [x^k] = \Delta_{\qB} c [x^k]$.
    \item If $x = \bot$
      then $\Delta_{\pB} = \Delta_{\qB} c$.
  \end{itemize}
\end{lemma}

%

\begin{algorithm}
  \caption{Simplifying a state $\pB = (x,\Delta,H)\in \Qbi$.}
  \label{simplify}
  \begin{algorithmic}[1]
    \algrestore{bkbreak}
    \Function{simplify}{$\pB$}
      \If{$\pB = (\bot, \Delta, H)$} \label{simplify-nc}
        \State \Let $\Delta' = \Delta \act H$,
        \ $c = \lcc(\Delta')$
        and $\qB = (\bot, \Delta' . c^{-1}, id_{\Delta'})$
        \State \Return $(\qB, c)$
      \ElsIf{$\pB = (x, \Delta, H_1 H_2 H_3)$,
           where $x \in \outL$ and  $H_2$ is the first loop in $H$,} \label{simplify-loop}
        \State \Let $\qB = (x, \Delta, H_1)$
        \State \Let $(\rB, c) = \Call{extend\_with\_loop}{\qB, H_2}$ with $\rB=(x',\Delta',H')$
        \State \Let $(\sB, d) = \Call{simplify}{(x',\Delta',H' . H_3)}$
        \State \Return $(\sB, dc)$
      \Else \label{simplify-last-else}
        \State \Return $(\pB, c_\eps)$
      \EndIf
    \EndFunction
  \end{algorithmic}
\end{algorithm}

\begin{theorem}\label{t:determinisation}
$\D$ is a finite sequential \longStoC{} equivalent to $\T$.
\end{theorem}

\begin{proof}[Proof Sketch]
First observe that $\D$ is sequential. The correctness of $\D$
is a consequence of the following property, that we prove using~\Cref{r:corr-simplify}
and an induction on $|u|$:
for all $u\in\inA^*$, if we have $\iB \ttrans{u}{c} \qB$ in $\D$,
then  $\Delta_{\qB} c [\eps] = (\tinit \act u) [\eps]$.
Last, we prove that $\D$ is finite. By construction, for every state
$\qB=(x,\Delta,H)$ of $\D$, $H$ contains no loop, hence its length is bounded
by $|Q|^{|Q|}$. This can be used to bound the size of $x$, as well
as the size of $\Delta$, for states in $\Qbs\cup\Qbc$.
The case of states in $\Qbnc$ is different: when such a state
$(\bot,\Delta,id_{\Delta})$ is
reached, then by the \twoloop{}, the transducers $\Lc{\T_\Delta}$ and
$\Rc{\T_\Delta}$ both satisfy the (classical)
twinning property. It remains to observe that the operations performed on Line~24
precisely correspond to two determinisations of~\cite{Choffrut77}, on both sides
of the \StoC{}.
\end{proof}

%% file: figure/sketch-states.tex
\begin{tikzpicture}[
  ->,
  >=stealth',
  shorten >=1pt,
  auto,
  node distance=1cm,
  scale=1,
  every node/.style={font=\small},
  asymcloud/.style={
    cloud, cloud ignores aspect,
    minimum width=3cm, minimum height=1.5cm,
    draw=black!40, very thick, fill=black!5,
    align=left,
  },
  dot/.style={fill, circle, inner sep=1pt},
  wordprefix/.style={
    draw,rectangle,red,inner sep=.5pt,text opacity=0
  },
  wordprefix crossed out/.style={
    draw,cross out,red,inner sep=-.5pt,
  },
  every edge/.style={
    draw=black,
    very thick,
    line cap=round,
  },
  off/.style={opacity=0},
  annotation/.style={
    rectangle,
    draw=black!100,
    fill=black!20
  },
]

  \node[asymcloud, cloud puffs=15] (cloud) at (0cm, 0cm) {\color{black}$\Qbs$};

  \node[asymcloud, cloud puffs=13] (cloud) at (5cm, 0cm) {\color{black}$\Qbc$};

  \node[asymcloud, cloud puffs=14] (cloud) at (10cm, 0cm) {\color{black}$\Qbnc$};

  \newcommand{\nprod}{$\neg \textsf{prod}$}
  \newcommand{\xcom}{$x\textsf{-com}$}
  \newcommand{\ncom}{$\neg \textsf{com}$}

  \path (-2,-.1) edge [out=20,in=160] (-.5,-.1);
  \path (.7,-.1) edge [out=10,in=170] node [above,sloped] {\xcom} (4.5,-.1);
  \path (.55,.1) edge [out=30,in=150] node [above,sloped] {\ncom} (9.45,.1);
  \path (5.5,-.1) edge [out=10,in=170] node [above,sloped] {\ncom} (9.3,-.1);

  \path (.3,.2) edge [loop left, distance=1.2cm, out=50, in=130]
    node [above,sloped,xshift=-.1cm,yshift=-.1cm] {\nprod} (-.3,.2);
  \path (5.3,.2) edge [loop left, distance=1.2cm, out=50, in=130]
    node [above,sloped,xshift=-.1cm,yshift=-.1cm] {\nprod\,$\vee$\! \xcom} (4.6,.2);
  \path (10.3,.2) edge [loop left, distance=1.2cm, out=50, in=130]
    node [above,sloped,xshift=-.1cm,yshift=-.1cm] {} (9.6,.2);

  \path (.5,-.3) edge [out=-25,in=180] (1.75,-.6);
  \path (5.5,-.3) edge [out=-25,in=180] (6.75,-.6);
  \path (10.5,-.3) edge [out=-25,in=180] (11.75,-.6);

\end{tikzpicture}

%% file: 070-decision.tex

\theoremstyle{plain}
\newtheorem*{problem}{Problem}

In this section, we prove the following result:
\begin{theorem}\label{r:decision-det}
  Given a \longStoC{}, determining whether there exists an equivalent
  sequential \longStoC{} is in \textsf{coNP}.
\end{theorem}


\newcommand{\smalltwoloop}{small-2-loop property}

In order to show this result, we introduce a restriction of the \twoloop{}:
%
\begin{definition}[\smalltwoloop]\label{small-2-loop}
  A \longStoC{} $\T$ is said to satisfy the \smalltwoloop{}
  if, for all runs $H_1,H_2,H_3,H_4 \in \T^2$ with
    $|H_i|\le |Q|^2$ for each $i$,
  $H_1H_2$, $H_1H_3H_4$ are lassos and they satisfy the \twoloop{} (in the sense of~\Cref{2-loop}).
%
\end{definition}


By definition,  if a \longStoC{} satisfies the \twoloop{}
then it also satisfies the \smalltwoloop{}. We will show that the two properties
are equivalent.
\begin{lemma}\label{r:small-2-loop-implies-lip}
  If a \longStoC{} $\T$ satisfies the \smalltwoloop{}
  then $\inter{\T}$ satisfies the contextual Lipschitz property.
\end{lemma}


\begin{proof}[Proof Sketch]
 We claim  there exists $K\in\N$ such that for every pair of synchronised runs
 $H:\ \ttrans{}{(c_0,d_0)} (p_0,q_0)  \ttrans{u}{(c_1,d_1)} (p_1,q_1)$ in $\T^2$, we have
 $\distf(c_1c_0[\eps],d_1d_0[\eps])\le K$. The result then easily follows.
 To prove this claim,
 we apply the main procedure \Call{simplify}{} (see~\Cref{sec:construction}) to
 the state $\pB=(\eps,\tinit,H)$.
 This procedure can indeed be applied: as it always processes the \emph{first} loop
 (see Line~29), the lassos considered satisfy the premises of the \smalltwoloop{}.
 The claim follows from the proof of finiteness of $\D$.
\end{proof}

%







\begin{proof}[Proof Sketch of \Cref{r:decision-det}]
  By \Cref{t:main} and~\Cref{r:small-2-loop-implies-lip}, $\T$ admits an
  equivalent sequential \StoC{} transducer iff $\T$ satisfies the
  \smalltwoloop{} (see also the figure below).
%
  Thus, we describe a procedure to decide
    whether $\T$ satisfies the
  \smalltwoloop{}.

  The procedure first non-deterministically guesses a counter-example to the \smalltwoloop{}
    and then verifies that it is indeed a counter-example.
  By definition of the \smalltwoloop{},
    the counter-example can have finitely many shapes.
  Those shapes require the verification of the properties of the involved lassos:
    being productive or not,
    being commuting or not,
    being aligned or not,
    satisfying the (classical) twinning property,
    etc.

  Verifying that a lasso in $\T^2$ is not commuting (resp. not aligned) boils down to checking
    whether there exists no $x \in \outLp$
    such that the lasso is $\Com{x}$ (resp. no $f \in \outC$ and $w \in \outL$
    such that the lasso is $\Align{f,w}$).
  %
  %
  In both cases, the search space for the words $x,w$ and context $f$
    can be narrowed down to factors of the output contexts of the given lasso.
  Thus these verifications can be done in polynomial time.
  The classical twinning property can also be checked in polynomial time.
  As a summary, we can show that
    the verifications for all the shapes
    can be done in polynomial time.
  Furthermore, all the shapes are of polynomial size, by definition of the
  \smalltwoloop{}, yielding the result.
  %
\end{proof}

%% file: 080-conclusion.tex


\begin{wrapfigure}{r}{.4\textwidth}
  \centering
  \input{figure/properties-sketch}
\end{wrapfigure}

We have proposed a multiple characterisation of \longStoC{s} that admit an
equivalent sequential \StoC{}, including a machine independent property,
 a pattern property, as well as a "small" pattern property
allowing to derive a decision procedure running in non-deterministic
polynomial time.
All these equivalences are summarised on the diagram on the right.
Future work includes a lower bound for the complexity of the problem, the extension of this
work to the register minimisation problem for streaming string transducers without
register concatenation, and the extension of our results to infinite words.


%% file: figure/properties-sketch.tex
%
\Crefname{theorem}{Thm.}{}
\Crefname{proposition}{Prop.}{}
\Crefname{lemma}{Lm.}{}
\begin{tikzpicture}[
  node distance=1.4cm,
  characterization/.style={font=\small,align=center,minimum width=1.5cm,minimum height=.5cm},
  tall/.style={minimum height=1.3cm},
  logic/.style={line width=.8pt,double equal sign distance},
  equivalence/.style={font=\small,Implies-Implies,logic},
  implication/.style={font=\small,-Implies,logic},
]
  \clip (-3,-4) rectangle (2.6, .25);

  \node[characterization] (LIP) {CLip};
  \node[characterization] (CTP) [below=of LIP] {CTP};
  \node[characterization] (2l) [below=of CTP] {2-loop};

  \node[characterization] (DET) [left=of CTP, xshift=1.4cm] {sequential \StoC{}};
  \node[characterization] (s2l) [right=of CTP, xshift=-1.4cm] {small-2-loop};

  \path (LIP) edge [implication] node [below,sloped] {\Cref{r:lip-implies-ctp}} (CTP);
  \path (CTP) edge [implication] node [below,sloped] {\Cref{r:ctp-implies-2-loop}} (2l);

  \path[transform canvas={xshift=-.4cm}]
    (2l) edge [implication] node [below,sloped] {Construction} (DET);
  \path[transform canvas={xshift=-.4cm}]
    (DET) edge [implication] node [above,sloped] {\Cref{r:det-implies-lip}} (LIP);

  \path[transform canvas={xshift=.4cm}]
    (2l) edge [implication] node [below,sloped] {Trivial} (s2l);
  \path[transform canvas={xshift=.4cm}]
    (s2l) edge [implication] node [above,sloped] {\Cref{r:small-2-loop-implies-lip}} (LIP);

\end{tikzpicture}

%% file: 120-annex-preliminaries.tex

\begin{proof}[Proof of \Cref{r:distf}]
  Let $x,y,z$ words.
  \begin{itemize}
  \item \textbf{Symmetry:} It is trivial to prove that $\distf(x,y) = \distf(y,x)$.
  \item \textbf{Identity:} It is trivial to prove that $\distf(x,y) = 0 \iff x = y$.
  \item \textbf{Triangle Inequality:}
    We want to prove that $\distf(x,z) \leq \distf(x,y) + \distf(y,z)$.
    By definition,
      $\distf(x,y) = |x| + |y| - 2|\lcf(x,y)|$,
      $\distf(y,z) = |y| + |z| - 2|\lcf(y,z)|$ and
      $\distf(x,z) = |x| + |z| - 2|\lcf(x,z)|$.
    Let $\alpha$ a longest common factor of $x$ and $y$,
    and $\beta$ a longest common factor of $y$ and $z$.
    Let $x_1, x_2, y_1, y_2, y_3, y_4, z_1, z_2$ words such that
      $x = x_1 \alpha x_2$, $y = y_1 \alpha y_2$, and
      $y = y_3 \beta y_4$, $z = z_1 \beta z_2$.
    Observe that
      $\distf(x,y) = |x_1| + |x_2| + |y_1| + |y_2|$ and
      $\distf(y,z) = |y_3| + |y_4| + |z_1| + |z_2|$.
    We observe six cases:
    \begin{enumerate}[label=(\roman*)]
      \item $|y_1| \leq |y_3| < |y_1| + |\alpha|$ and $|y_4| \leq |y_2| < |y_4| + |\beta|$.

        There exists $\gamma$ such that
          $y_1 \alpha = y_3 \gamma$ and
          $\gamma y_2 = \beta y_4$.
        Then we have
          $|\alpha| \leq |y_3| + |\gamma|$ and
          $|\beta| \leq |y_2| + |\gamma|$.
        Yet, $|x| = |x_1| + |x_2| + |\alpha|$ and we obtain that
          $|x| - |\gamma| = |x_1| + |x_2| + |\alpha| - |\gamma| \leq |x_1| + |x_2| + |y_3|$.
        Also, $|z| = |z_1| + |z_2| + |\beta|$ and we obtain that
          $|z| - |\gamma| = |z_1| + |z_2| + |\beta| - |\gamma| \leq |z_1| + |z_2| + |y_2|$.
        Finally, we have that $|\lcf(x,z)| \geq |\gamma|$
          as $\gamma$ indeed is a common factor of $x$ and $z$.
        Then,
          \begin{align*}
            \distf(x,z) & = |x| + |z| - 2|\lcf(x,z)| \\
                        & \leq |x| + |z| - 2|\gamma| \\
                        & \leq |x_1| + |x_2| + |y_3| + |z_1| + |z_2| + |y_2| \\
                        & \leq \distf(x,y) + \distf(y,z)
          \end{align*}

      \item $|y_3| \leq |y_1| < |y_3| + |\alpha|$ and $|y_2| \leq |y_4| < |y_2| + |\beta|$.

        This case is symmetrical to the previous case.

      \item $|y_1| \leq |y_3| < |y_1| + |\alpha|$ and $|y_2| \leq |y_4| < |y_2| + |\beta|$.

        Then we have
          $|\alpha| \leq |y_3| + |\beta| + |y_4|$.
        Yet, $|x| = |x_1| + |x_2| + |\alpha|$ and we obtain that
          $|x| - |\beta| = |x_1| + |x_2| + |\alpha| - |\beta| \leq |x_1| + |x_2| + |y_3| + |y_4|$.
        Also, $|z| = |z_1| + |z_2| + |\beta|$ and thus $|z| - |\beta| = |z_1| + |z_2|$.
        Finally, we have that $|\lcf(x,z)| \geq |\beta|$
          as $\beta$ indeed is a common factor of $x$ and $z$.
        Then,
          \begin{align*}
            \distf(x,z) & = |x| + |z| - 2|\lcf(x,z)| \\
                        & \leq |x| + |z| - 2|\beta| \\
                        & \leq |x_1| + |x_2| + |y_3| + |y_4| + |z_1| + |z_2| \\
                        & \leq \distf(x,y) + \distf(y,z)
          \end{align*}

      \item $|y_3| \leq |y_1| < |y_3| + |\alpha|$ and $|y_4| \leq |y_2| < |y_4| + |\beta|$.

        This case is symmetrical to the previous case.

      \item $|y_1| + |\alpha| \leq |y_3|$ and $|y_4| + |\beta| \leq |y_2|$.

        Then, $|\alpha| \leq |y_3|$ and $|\beta| \leq |y_2|$.
        Yet, $|x| = |x_1| + |x_2| + |\alpha|$ and we obtain that
          $|x| \leq |x_1| + |x_2| + |y_3|$.
        Also, $|z| = |z_1| + |z_2| + |\beta|$ and we obtain that
          $|z| \leq |z_1| + |z_2| + |y_2|$.
        Then, as $|\lcf(x,z)| \geq 0$,
          \begin{align*}
            \distf(x,z) & = |x| + |z| - 2|\lcf(x,z)| \\
                        & \leq |x| + |z| \\
                        & \leq |x_1| + |x_2| + |y_3| + |z_1| + |z_2| + |y_2| \\
                        & \leq \distf(x,y) + \distf(y,z)
          \end{align*}

      \item $|y_3| + |\alpha| \leq |y_1|$ and $|y_2| + |\beta| \leq |y_4|$.

        This case is symmetrical to the previous case.
      \qedhere
    \end{enumerate}
  \end{itemize}
\end{proof}


\begin{lemma}
  Let $c,c' \in \outC$ and $w,w' \in \outL$.
  $\distf(c[w],c'[w']) \leq \distf(w,w') + |c| + |c'|$.
\end{lemma}

\begin{proof}
  It is easy to see that $|\lcf(c[w], c'[w'])| \geq |\lcf(w,w')|$.
  Then $|c| + |c'| + |w| + |w'| -2|\lcf(c[w], c'[w'])|
          \leq |c| + |c'| + |w| + |w'| -2|\lcf(w,w')|$.
  And we obtain the result.
\end{proof}


\begin{corollary}
  Let $c,c' \in \outC$ and $w \in \outL$.
  $\distf(c[w],c'[w]) \leq |c| + |c'|$.
\end{corollary}


\begin{lemma}
  Let $c,c' \in \outC$ and $w,w' \in \outL$.
  $\distf(w,w') \leq \distf(c[w],c'[w']) + |c| + |c'|$.
\end{lemma}

\begin{proof}
First observe that for any common factor $y$ of $w$ and $w'$, we have
 $\distf(w,w') \leq |w|+|w'|-2|y|$.
 To prove the result, consider a longest common factor
$x$ of $c[w]$ and $c'[w']$. We
claim that there exists a common factor $y$ of $w$ and $w'$ (which may be the empty word)
such that $|y|\geq |x|-|c|-|c'|$. Indeed, one obtains $y$ from $x$
by removing symbols of $x$ that correspond to positions in $c$ or $c'$.
As a consequence, we obtain:
$$
\begin{array}{lll}
\distf(w,w') & \leq & |w|+|w'|-2|y|\\
                  & \leq & |w|+|w'|-2|x|+2|c|+2|c'|\\
                  & \leq & \distf(c[w],c'[w'])  + |c| + |c'|
\end{array}
$$
This concludes the proof.
\end{proof}

%% file: 150-annex-combinatorics.tex

Throughout this section,
  we consider a \longStoC{} $\T=(Q, \tinit, \tfinal, T)$
  that satisfies the contextual twinning property (CTP).
Let $L \in \N$ such that any two states of $\T$ are $L$-contextually twinned.

\subsection{Additional Word Combinatorics Notations}

The size of a word $x$ is denoted by $|x|$.
Given two words $x,y \in \alphabet^*$,
  we write $x \preceq_p y$, resp. $x \preceq_s y$,
  if $x$ is a prefix, resp. a suffix, of $y$.
If we have $x \preceq_p y$, resp. $x \preceq_s y$,
  then we note $x^{-1}y$, resp. $yx^{-1}$,
  the unique word $z$ such that
  $y = xz$, resp. $y = zx$.
A word $x \in \alphabet^*$ is \intro{primitive}
  if there is no word $y$ such that $|y| < |x|$ and $x \in y^*$.
The \intro{primitive root} of a word $x \in \alphabet^*$,
  denoted by $\proot(x)$,
  is the (unique) primitive word $y$
  such that $x \in y^*$.
In particular, if $x$ is primitive, then its primitive root is $x$.
The \intro{primitive period} of a word $x \in \alphabet^*$,
  denoted by $\prootf(x)$,
  is the (unique) primitive word $y$
  such that $x \in y^+z$ for some $z \preceq_p y$.
Two words $x$ and $y$ are \intro{conjugates}, written $x \sim y$,
  if there exists $z \in \alphabet^*$ such that $xz=zy$.
It is well-known that two words are conjugates
  iff there exist $t_1,t_2 \in \alphabet^*$
  such that $x = t_1 t_2$ and $y = t_2 t_1$.

\begin{example}
  The primitive root and primitive period act differently.
  For instance, $\proot(abcab)=abcab$ but $\prootf(abcab)=abc$.
\end{example}

For $n,m \in \Nplus$, we note by $\gcd(n,m)$ the greatest common divisor of $n$ and $m$.


\begin{lemma}[Fine and Wilf,
    \cite{fine_uniqueness_1965},
    Chapter 9 of \cite{lothaire_algebraic_2002}%
  ]\label{r:fw}
  Let $x,y \in \alphabet^*$ and $m,n \in \N$.
  If $x^m$ and $y^n$ have a common subword of length at least $|x| + |y| - \gcd(|x|,|y|)$,
    then their primitive roots are conjugates.
\end{lemma}


\begin{lemma}[Saarela, Theorem 4.3 of \cite{saarela_systems_2015}]\label{r:saarela}
  Let $m,n \geq 1$, $s_j,t_j \in \alphabet^*$ and $u_j,v_j \in \alphabet^+$.
  If $s_0 u_1^i s_1 \dots u_m^i s_m = t_0 v_1^i t_1 \dots v_n^i t_n$ holds
    for $m+n$ values of $i$,
    then it holds for all $i$.
\end{lemma}


\begin{lemma}\label{r:prootf-of-factor}
  Let $x, y \in \inL$ such that $|y| \geq |x|$.
  If $y$ is a factor of $x^*$
    then $\prootf(y) \sim \proot(x)$.
\end{lemma}

\begin{proof}
  Without loss of generality, consider $x$ to be primitive.
  Let $u,v \in \inL$ and $i \in \Nplus$ such that $uyv = x^i$.
  There exist $t_1,t_1'$
    such that $x = t_1 t_1'$
    and $u \in (t_1 t_1')^* t_1$.
  Then $yv \in t_1' (t_1 t_1')^+$, as $|y| \geq |t_1' t_1|$.
  There exist $t_2,t_2'$
    such that $t_1' t_1 = t_2 t_2'$
    and $y \in (t_1' t_1)^+ t_2$.
  Thus, by definition of $\prootf(y)$,
    we have $\prootf(y) \sim x$.
\end{proof}

\subsection{Lassos}

Two lassos
  $\ttrans{}{c_1} p_1 \ttrans{u_1}{d_1} q_1 \ttrans{v_1}{e_1} q_1$ and
  $\ttrans{}{c_2} p_2 \ttrans{u_2}{d_2} q_2 \ttrans{v_2}{e_2} q_2$
are said to be \intro{weakly balanced} if $|e_1| = |e_2|$.

The following Lemma states a consequence of the definition of aligned lassos.

\begin{lemma}\label{r:aligned-means-conjugated}
  Let $\rho$ be a productive lasso $\ttrans{}{c} p \ttrans{u}{d} q \ttrans{v}{e} q$
  and $f \in \outC$ and $w \in \outL$.
  $\rho$ is $\Align{f,w}$ if and only if
    there exist $t_1,t_2,t_3,t_4$ and $\alpha,\beta \geq 0$ such that
      $\proot(\Lc{e}) = t_1 t_2$, $\proot(\Lc{f}) = t_2 t_1$,
      $\proot(\Rc{e}) = t_3 t_4$, $\proot(\Rc{f}) = t_4 t_3$, and
      $dc[\eps] = (t_1 t_2)^\alpha t_1 w t_3 (t_4 t_3)^\beta$.
\end{lemma}

As a corollary, any productive lasso $\rho :\; \ttrans{}{c} p \ttrans{u}{d} q \ttrans{v}{e} q$
  is $\Align{e,dc[\eps]}$.
Also, note that a lasso can be commuting and aligned at the same time.


The following lemma states the combinatorial properties of two synchronised lassos.

\begin{lemma}\label{r:ctp-concrete}
  For any two synchronised lassos $\rho_1$ and $\rho_2$,
  we have that
  \begin{itemize}
    \item either $\rho_1$ and $\rho_2$ are non-productive
    \item or
      $\rho_1$ and $\rho_2$ are productive and weakly-balanced,
      and there exists $x \in \outA^+$ primitive
      such that $\rho_1$ and $\rho_2$ are $\Com{x}$,
    \item or
      $\rho_1$ and $\rho_2$ are productive, strongly-balanced and non-commuting,
      and there exists $f \in \outC$ and $w \in \outL$
      such that $\rho_1$ and $\rho_2$ are $\Align{f,w}$.
  \end{itemize}
\end{lemma}


In order to prove \Cref{r:ctp-concrete}, we first need some preliminary combinatorial results.

\begin{lemma}\label{r:dist-to-equation}
  Let $c_1,c_2,d_1,d_2 \in \Contexts{B}$.
  If for all $i \in \N$,
      $\distf(d_1^i c_1 [\eps], d_2^i c_2 [\eps]) \leq L$,
  then there exist $e_1,e_2 \in \Contexts{B}$
    such that for all $i \in \N$,
      $e_1 d_1^i c_1 [\eps] = e_2 d_2^i c_2 [\eps]$.
\end{lemma}

\begin{proof}
  Suppose that for all $i \in \N$,
    $\distf(d_1^i c_1 [\eps],d_2^i c_2 [\eps]) \leq L$.
  Then for all $i \in \N$,
    there exist $f_1,f_2 \in \Contexts{B}$,
    such that $f_1^{-1} d_1^i c_1 [\eps] = f_2^{-1} d_2^i c_2 [\eps]$
    and $|f_1| + |f_2| \leq L$.
  Let $C_L = \{ (f_1,f_2) \mid |f_1| + |f_2| \leq L \}$.
  $C_L$ is finite.
  Thus there exist some $(f_1,f_2) \in C_L$
    such that there is an infinite number of $i \in \N$
      such that $f_1^{-1} d_1^i c_1 [\eps] = f_2^{-1} d_2^i c_2 [\eps]$.
  Furthermore, there exists $i_0$ such that for all $i \geq i_0$,
    there exists $e_1,e_2 \in \Contexts{B}$
    and $e_1 d_1^{i-i_0} c_1 [\eps] = e_2 d_2^{i-i_0} c_2 [\eps]$.
  Finally, by \Cref{r:saarela}, we obtain that
    for all $i \in \N$,
      $e_1 d_1^i c_1 [\eps] = e_2 d_2^i c_2 [\eps]$.
\end{proof}


\begin{lemma}\label{r:comb-1-1}
  Let $u_1,w_1,u_2,w_2 \in \outA^*$ and $v_1,v_2 \in \outA^+$
    such that $|v_1| = |v_2|$
    and $u_1 v_1^i w_1 = u_2 v_2^i w_2$ for all $i \in \N$.
  Then there exists $x \in \outLp$ and $f_1,f_2 \in \outC$ such that
    for all $i \geq 1$, there exist $k \in \N$ such that
      $v_1^i = f_1 [x^k]$ and $v_2^i = f_2 [x^k]$.
\end{lemma}

\begin{proof}
  There exists $i_0$ sufficiently large so that $v_1^{i_0}$ and $v_2^{i_0}$ overlap
    with a common factor of length greater than $|v_1| + |v_2| - \gcd(|v_1|, |v_2|)$.
  Thus, by \Cref{r:fw}, $\proot(v_1) \sim \proot(v_2)$.

  Let $t,t' \in \outL$ and $\alpha,\beta \geq 1$
    such that $v_1 = (t t')^\alpha$ and $v_2 = (t' t)^\beta$.
  We choose $x = \rho(v_1) = t t'$,
    $f_1 = (t t', \eps)$ and $f_2 = (t', t)$.
  Then for all $i \geq 1$,
    let $k = \alpha i - 1 \geq 0$,
    and we have $v_1^i
      = (t t')^{\alpha i}
      = t t' (t t')^{\alpha i - 1}
      = f_1 [x^k]$
    and $v_2^i
      = (t' t)^{\alpha i}
      = t' x^{\alpha i - 1} t
      = f_2 [x^k]$.
\end{proof}

\begin{lemma}\label{r:comb-1-2}
  Let $u_1,w_1,u_2,w_2,y_2 \in \outL$ and $v_1,v_2,x_2 \in \outA^+$
    such that $|v_1| = |v_2| + |x_2|$
    and $u_1 v_1^i w_1 = u_2 v_2^i w_2 x_2^i y_2$ for all $i \in \N$.
  Then there exists $x \in \outLp$ and $f_1,f_2 \in \outC$ such that
    for all $i \geq 1$, there exist $k \in \N$ such that
      $v_1^i = f_1 [x^k]$ and $v_2^i w_2 x_2^i = f_2 [x^k]$.
\end{lemma}

\begin{proof}
  There exists $m_0$ sufficiently large so that $v_1^{m_0}$
    overlap with both $v_2^{m_0}$ and $x_2^{m_0}$
    with common factors of length greater than
      $|v_1| + |v_2| - \gcd(|v_1|, |v_2|)$ and $|v_1| + |x_2| - \gcd(|v_1|, |x_2|)$.
  Thus, by \Cref{r:fw},
    $\proot(v_1) \sim \proot(v_2)$ and $\proot(v_1) \sim \proot(x_2)$.
  As $|v_1| = |v_2| + |x_2|$,
    we have that $|v_2 w_2 x_2| \geq |v_1|$.
  Yet $v_2 w_2 x_2$ is a factor of $v_1^*$,
    then, by \Cref{r:prootf-of-factor},
      $\prootf(v_2 w_2 x_2) \sim \proot(v_1)$.

  Let $t_1,t_1',t_2,t_2' \in \outL$ and $\alpha,\beta,\gamma \geq 1$
    such that $t_1 t_1' = t_2 t_2'$,
    and $v_1 = (t_1 t_1')^\alpha$,
        $v_2 = (t_1' t_1)^\beta$
    and $x_2 = (t_2' t_2)^\gamma$.
  Note that $\alpha = \beta + \gamma$.
  Also we have $v_2 w_2 x_2 = t_1' (t_1 t_1')^\theta t_2$, for some $\theta \geq 0$.
  We choose $x = \rho(v_1) = t_1 t_1'$,
    $f_1 = ( (t_1 t_1')^\alpha, \eps)$ and $f_2 = (t_1' (t_1 t_1')^\theta, t_2)$.
  Then for all $i \geq 1$,
    let $k = \alpha (i-1) \geq 0$,
    and we have  $v_1^i
      = (t_1 t_1')^{\alpha i}
      = (t_1 t_1')^\alpha x^{\alpha (i-1)}
      = f_1 [x^k]$
    and $v_2^i w_2 x_2^i
      = v_2^{i-1} v_2 w_2 x_2 x_2^{i-1}
      = (t_1' t_1)^{\beta (i-1)} t_1' (t_1 t_1')^\theta t_2 (t_2' t_2)^{\gamma (i-1)}
      = t_1' (t_1 t_1')^\theta x^{\alpha (i-1)} t_2
      = f_2 [x^k]$.
\end{proof}

\begin{lemma}\label{r:comb-2-2}
  Let $u_1,w_1,y_1,u_2,w_2,y_2 \in \outA^*$ and $v_1,x_1,v_2,x_2 \in \outA^+$
    such that $|v_1| + |x_1| = |v_2| + |x_2|$
    and $u_1 v_1^i w_1 x_1^i y_1 = u_2 v_2^i w_2 x_2^i y_2$
      for all $i \in \N$.
  \begin{itemize}
    \item
      If $|v_1| \neq |v_2|$ and $|x_1| \neq |x_2|$
      then there exists $x \in \outLp$ and $f_1,f_2 \in \outC$ such that
        for all $i \geq 1$, there exist $k \in \N$ such that
          $v_1^i w_1 x_1^i = f_1 [x^k]$ and $v_2^i w_2 x_2^i = f_2 [x^k]$.
    \item
      If $|v_1| = |v_2|$ and $|x_1| = |x_2|$
      then there exist $w \in \outL$ and $f,g_1,g_2 \in \outC$ such that
        for all $i \in \N$,
          $v_1^i w_1 x_1^i = g_1 f^i [w]$ and $v_2^i w_2 x_2^i = g_2 f^i [w]$.
  \end{itemize}
\end{lemma}

\begin{proof}
  If $|v_1| \neq |v_2|$ and $|x_1| \neq |x_2|$, suppose $|v_1| > |v_2|$.
  There exists $i_0$ sufficiently large so that $v_1^{i_0}$
    overlap with both $x_2^{i_0}$
    with a common factor of length greater than
      $|v_1| + |x_2| - \gcd(|v_1|, |x_2|)$.
  Thus, by \Cref{r:fw}, $\proot(v_1) \sim \proot(x_2)$.
  Using the same argument,
    we have that $\proot(v_1) \sim \proot(v_2)$ and $\proot(x_1) \sim \proot(x_2)$.
  As $|v_1| + |x_1| = |v_2| + |x_2|$,
    we have that $|v_2| + |x_2| \geq |v_1|$
    and thus $|v_2 w_2 x_2| \geq |v_1|$.
  Yet $v_2 w_2 x_2$ is a factor of $v_1^*$,
    then, by \Cref{r:prootf-of-factor},
      $\prootf(v_2 w_2 x_2) \sim \proot(v_1)$.
  Symmetrically, $\prootf(v_1 w_1 x_1) \sim \proot(x_2)$.

  Let $t_1,t_1',t_2,t_2',t_3,t_3' \in \outL$ and $\alpha,\beta,\gamma,\delta \geq 1$
    such that $t_1 t_1' = t_2 t_2' = t_3 t_3'$,
    and $v_1 = (t_1 t_1')^\alpha$, $x_1 = (t_2' t_2)^\beta$,
      $v_2 = (t_1' t_1)^\gamma$ and $x_2 = (t_3' t_3)^\delta$.
  Note that $\alpha + \beta = \gamma + \delta$.
  Also we have
    $v_1 w_1 x_1 = (t_1 t_1')^{\theta_1} t_2$
    and $v_2 w_2 x_2 = t_1' (t_1 t_1')^{\theta_2} t_3$,
    for some $\theta_1,\theta_2 \geq 0$.
  We choose $x = \rho(v_1) = t_1 t_1'$
    $f_1 = ((t_1 t_1')^{\theta_1}, t_2)$
    and $f_2 = (t_1' (t_1 t_1')^{\theta_2}, t_3)$.
  Then for all $i \geq 1$,
    let $k = (\alpha + \beta) (i-1) \geq 0$,
    and we have
      $v_1^i w_1 x_1^i
      = v_1^{i-1} v_1 w_1 x_1 x_1^{i-1}
      = (t_1 t_1')^{\theta_1} (t_1 t_1')^{(\alpha + \beta) (i-1)} t_2
      = f_1 [x^k]$
    and
      $v_2^i w_2 x_2^i
      = v_2^{i-1} v_2 w_2 x_2 x_2^{i-1}
      = t_1' (t_1 t_1')^{\theta_2} (t_1 t_1')^{(\gamma + \delta) (i-1)} t_3
      = f_2 [x^k]$.

  If $|v_1| < |v_2|$, we obtain the same result.
  \medskip

  If $|v_1| = |v_2|$ and $|x_1| = |x_2|$,
    we only have that $\proot(v_1) \sim \proot(v_2)$ and $\proot(x_1) \sim \proot(x_2)$.
  If $|u_1| < |u_2|$, let $v$ such that we have $u_2 = u_1 v$ and $v_1 v = v v_2$;
  if $|u_1| = |u_2|$, let $v = \varepsilon$ and we have $u_1 = u_2$ and $v_1 = v_2$;
  if $|u_1| > |u_2|$, let $v$ such that we have $u_1 = u_2 v$ and $v_2 v = v v_1$.
  Similarly,
  if $|y_1| < |y_2|$, let $x$ such that we have $y_2 = x y_1$ and $x_2 x = x x_1$;
  if $|y_1| = |y_2|$, let $x = \varepsilon$ and we have $y_1 = y_2$ and $x_1 = x_2$;
  if $|y_1| > |y_2|$, let $x$ such that we have $y_1 = x y_2$ and $x_1 x = x x_2$.

  Finally, from $v$ and $y$, we obtain that
  \begin{itemize}
    \item $v_1 v = v v_2$, $w_1 = v w_2 x$, $x_2 x = x x_1$, or
    \item $v_1 v = v v_2$, $w_1 x = v w_2$, $x_1 x = x x_2$, or
    \item $v_2 v = v v_1$, $v w_1 = w_2 x$, $x_2 x = x x_1$, or
    \item $v_2 v = v v_1$, $v w_1 x = w_2$, $x_1 x = x x_2$.
  \end{itemize}

  We handle the first case.
    The others are similar.
  We choose $f = (v_2,x_2)$, $w = w_2$, and
    $g_1 = (v, x)$ and $g_2 = (\eps, \eps)$.
  Then for all $i \in \N$,
    $v_1^i w_1 x_1^i
      = v_1^i w_1 x_1^i
      = v_1^i v w_2 x x_1^i
      = v v_2^i w_2 x_2^i x
      = v (f^i [w]) x
      = g_1 f^i [w]$
    and $v_2^i w_2 x_2^i = g_2 f^i [w]$.
\end{proof}


\begin{lemma}\label{r:aligned-can-be-com}
  Let $f \in \outC$, $w \in \outL$, and $x \in \outLp$ a primitive word.
  Let $\rho_1$ and $\rho_2$ be two
    synchronised, productive, strongly-balanced and $\Align{f,w}$ lassos.
  If $\rho_1$ is $\Com{x}$, then $\rho_2$ is $\Com{x}$.
\end{lemma}

\begin{proof}
  Let $\rho_1:\; \ttrans{}{c_1} p_1 \ttrans{u_1}{d_1} q_1 \ttrans{u_2}{e_1} q_1$
  and $\rho_2:\; \ttrans{}{c_2} p_2 \ttrans{u_1}{d_2} q_2 \ttrans{u_2}{e_2} q_2$.

  We have that $\|e_1\| = \|e_2\|$,
    and $\rho_1$ and $\rho_2$ are $\Align{f,w}$.
  By \Cref{d:aligned-lasso},
    there exist some contexts $g_1,g_2 \in \outC$
    such that
      for all $i \geq 0$,
        $e_1^i d_1 c_1 [\eps] = g_1 f^i w$
        and $e_2^i d_2 c_2 [\eps] = g_2 f^i w$.
  If $\rho_1$ is $\Com{x}$ then, by \Cref{d:commuting-lasso},
    there exist $h \in \outC$ such that
      for all $i \geq 0$, there exists $j \geq 0$ such that
        $e_1^i d_1 c_1 [\eps] = h [x^j]$.
      Hence, there exists $k \geq 0$ and $h' \in \outC$ such that
        $f^i w = h' [x^k]$
        and then $g_2 f^i [w] = g_2 h' [x^k]$.
  Therefore, by \Cref{d:commuting-lasso}, $\rho_2$ is $\Com{x}$.
\end{proof}


We can now prove \Cref{r:ctp-concrete}.

\begin{proof}[Proof of \Cref{r:ctp-concrete}]
  Let $\rho_1:\; \ttrans{}{c_1} p_1 \ttrans{u_1}{d_1} q_1 \ttrans{u_2}{e_1} q_1$
  and $\rho_2:\; \ttrans{}{c_2} p_2 \ttrans{u_1}{d_2} q_2 \ttrans{u_2}{e_2} q_2$.
  By \Cref{r:dist-to-equation},
    there exist $f_1,f_2 \in \outC$ such that for all $i \in \N$
      $f_1 e_1^i d_1 c_1 [\eps] = f_2 e_2^i d_2 c_2 [\eps]$.
  Then we have that $|e_1| = |e_2|$.
  We observe 10 cases.
  \medskip

  If $|e_1| = 0$ or $|e_2| = 0$ then $|e_1| = |e_2| = 0$ and $\rho_1,\rho_2$ are not productive.
  \medskip

  If $e_1, e_2 \in \outLp \times \{\eps\}$,
    then by \Cref{r:comb-1-1},
    there exists $x \in \outLp$ such that
      both $\rho_1$ and $\rho_2$ are productive, weakly-balanced and $\Com{x}$.
  The same holds for the other three cases
    where exactly two of the four components of $e_1$ and $e_2$ are empty.
  \medskip

  If $e_1 \in \outLp \times \{\eps\}$ and $e_2 \in \outLp \times \outLp$,
    then by \Cref{r:comb-1-2},
    there exists $x \in \outLp$ such that
      both $\rho_1$ and $\rho_2$ are productive, weakly-balanced and $\Com{x}$.
  The same holds for the other three cases
    where exactly one of the four components of $e_1$ and $e_2$ is empty.
  \medskip

  If $e_1, e_2 \in \outLp \times \outLp$,
    then by \Cref{r:comb-2-2},
    there are two cases.
  Firstly, if $\|e_1\| \neq \|e_2\|$
  then there exists $x \in \outLp$ such that
    both $\rho_1$ and $\rho_2$ are productive, weakly-balanced and $\Com{x}$.
  Secondly, if $\|e_1\| = \|e_2\|$
  then there exist $f \in \outC$ and $w \in \outL$ such that
    both $\rho_1$ and $\rho_2$ are productive, strongly-balanced, and $\Align{f,w}$.
  However, by \Cref{r:aligned-can-be-com},
    if it still happens that either one of $\rho_1$ and $\rho_2$ is $\Com{x}$,
    then both $\rho_1$ and $\rho_2$ are $\Com{x}$.
  If not, then they both are non-commuting.
\end{proof}


\begin{lemma}\label{r:conjugated-com}
  Let $\rho$ be a productive lasso,
  and $x,x' \in \outA^+$ be two primitive words.
  \begin{itemize}
    \item If $x \sim x'$ and $\rho$ is $\Com{x}$ then $\rho$ is $\Com{x'}$
    \item If $\rho$ is $\Com{x}$ and $\Com{x'}$ then $x \sim x'$
  \end{itemize}
\end{lemma}

\begin{proof}
  Let $\rho:\; \ttrans{}{c} p \ttrans{u}{d} q \ttrans{u}{e} q$
    and $x,x' \in \outA^+$ be two primitive words.

  Firstly, suppose that $x \sim x'$ and $\rho$ is $\Com{x}$.
  By definition, there exists $f \in \outC$ such that
    for all $i \in \N$ there exists $k \in \N$ such that
      $e^i d c [\eps] = f [x^k]$.
  As $x \sim x'$, there exist $t,t' \in \outL$ such that $x = tt'$ and $x' = t't$.
  Let $g = (t, t')$.
  We obtain that
    for all $i \in \N$ there exists $k \in \N$ such that
      $e^i d c [\eps] = f [x^k] = f g [x'^{k-1}]$.

  Secondly, suppose that $\rho$ is both $\Com{x}$ and $\Com{x'}$.
  By definition, there exists $f,f' \in \outC$ such that
    for all $i \in \N$ there exists $k,k' \in \N$ such that
      $e^i d c [\eps] = f [x^k] = f' [x'^{k'}]$.
  As $x$ and $x'$ are primitive, we have that $k=k'$ and thus $|x| = |x'|$.
  By \Cref{r:comb-1-1}, we obtain that $x \sim x'$.
\end{proof}


\begin{lemma}\label{r:com-means-com}
  Let $x \in \outLp$ a primitive word.
  Let $\rho_1$ and $\rho_2$ be two synchronised productive lassos.
  If $\rho_1$ is $\Com{x}$, then $\rho_2$ is $\Com{x}$.
\end{lemma}

\begin{proof}
  Let $\rho_1:\; \ttrans{}{c_1} p_1 \ttrans{u_1}{d_1} q_1 \ttrans{u_2}{e_1} q_1$
  and $\rho_2:\; \ttrans{}{c_2} p_2 \ttrans{u_1}{d_2} q_2 \ttrans{u_2}{e_2} q_2$.
  By \Cref{r:ctp-concrete}, we observe two cases.
  First, consider that $\|e_1\| = \|e_2\|$
    and that there exists $f \in \outC$ and $w \in \outL$
    such that $\rho_1$ and $\rho_2$ are $\Align{f,w}$.
  By \Cref{r:aligned-can-be-com},
    if $\rho_1$ is $\Com{x}$, then $\rho_2$ is $\Com{x}$.
  Otherwise, we only have that $|e_1| = |e_2|$,
    and there exists $x' \in \outA^+$ primitive
      such that $\rho_1$ and $\rho_2$ are $\Com{x'}$.
  If $\rho_1$ is $\Com{x}$
    then, by \Cref{r:conjugated-com}, we have that $x' \sim x$
    and, by \Cref{r:conjugated-com} again, that $\rho_2$ is $\Com{x}$.
\end{proof}


\begin{lemma}\label{r:all-correctly-aligned}
  Let $\rho_1,\dots,\rho_k$ be $k$ synchronised productive lassos
  that are pairwise aligned, strongly balanced and not commuting.
  Then there exist $f \in \outC$ and $w \in \outL$
  such that they are all $\Align{f,w}$.
\end{lemma}

\begin{proof}
  Let $\rho_i:\; \ttrans{}{c_i} p_i \ttrans{u_1}{d_i} q_i \ttrans{u_2}{e_i} q_i$
    for $i \in \{1,\dots,k\}$.
  As $\rho_1,\dots,\rho_k$ are pairwise aligned,
    there exist $f_2,\dots,f_k \in \outC$ and $w_2,\dots,w_k \in \outL$ such that
    for all $i \in \{2,\dots,k\}$, $\rho_1$ and $\rho_i$ are $\Align{f_i, w_i}$.
  %
  Then for all $i \in \{2,\dots,k\}$,
    there exist $f_i, g_i, h_i \in \outC$ and $w_i \in \outL$ such that
    for all $j \in \N$,
      $e_1^j d_1 c_1 [\eps] = g_i f_i^j [w_i]$
      and $e_i^j d_i c_i [\eps] = h_i f_i^j [w_i]$.

  Let $\ell, r \in \{2,\dots,k\}$ such that
    $|\Lc{g_\ell}| = max \{ |\Lc{g_i}| \mid i \in \{2,\dots,k\} \}$
    and $|\Rc{g_r}| = max \{ |\Rc{g_i}| \mid i \in \{2,\dots,k\} \}$.
  Let $g = (\Lc{g_\ell}, \Rc{g_r})$,
    $f = (\Lc{f_\ell}, \Rc{f_r})$,
    and $w = g^{-1} d_1 c_1 [\eps]$.
  By definition,
    for all $i \in \{2,\dots,k\}$,
      $|\Lc{g_i}| \leq |\Lc{g_\ell}|$ and $|\Lc{g_i}| \leq |\Lc{g_r}|$.
  Thus, for all $i \in \{2,\dots,k\}$, $g_i^{-1} g \in \outC$.
  %
  We have that $|g| > |d_1 c_1|$,
    otherwise it would contradict that the lassos are all non-commuting.
  Thus, $w \in \outL$.

  By \Cref{r:aligned-means-conjugated}, we have that
    $\proot(\Lc{e}) \sim \proot(\Lc{f_\ell})$ and
    $\proot(\Rc{e}) \sim \proot(\Rc{f_r})$.
  Therefore, we can show that
    for all $j \in \N$, $e_1^j d_1 c_1 [\eps] = g f^j [w]$.
  Then, for all $i \in \{2,\dots,k\}$ and $j \in \N$,
    $e_i^j d_i c_i [\eps]
      = h_i f_i^j [w_i]
      = h_i g_i^{-1} e_1^j d_1 c_1 [\eps]
      = h_i g_i^{-1} g f^j [w]$.
\end{proof}



\begin{proof}[Proof of \Cref{r:all-commuting-or-aligned}]
  The length of the contexts labelling the loops must be equal,
    as the outputs must grow at the same pace when the loops are pumped.
  By \Cref{r:com-means-com}, if one of the lassos is $\Com{x}$ then they are all $\Com{x}$.
  Otherwise, none of them are commuting.
  Then, by \Cref{r:ctp-concrete}, they are also all pairwise aligned and strongly balanced.
  Therefore by \Cref{r:all-correctly-aligned},
    there exists $f \in \outC$ and $w \in \outL$
    such that they are all $\Align{f,w}$.
\end{proof}


\subsection{Lassos Consecutive to a Commuting Lasso}


We can now state the following Lemma.

\begin{lemma}\label{r:strongly-commuting-or-aligned}
  Let $x \in \outLp$ a primitive word
  and let $\Delta = \split_c(x, H_1, H_2)$
    for some $H_1 H_2$ an $\Com{x}$ lasso in $\T^k$.
  For any two synchronised lassos $\rho_1$ and $\rho_2$ in $\T_{\Delta}$,
  we have that
  \begin{itemize}
    \item either $\rho_1$ and $\rho_2$ are non-productive,
    \item or
      $\rho_1$ and $\rho_2$ are productive, weakly-balanced, and $\SCom{x}$,
    \item or
      $\rho_1$ and $\rho_2$ are productive, strongly-balanced, non-commuting,
      and there exists $g,f \in \outC$
      such that $\rho_1$ and $\rho_2$ are $\SAlign{g,f,x}$.
  \end{itemize}
\end{lemma}


In order to prove \Cref{r:strongly-commuting-or-aligned},
  we first need some additional combinatorial results.


\begin{lemma}\label{r:com-dist-to-equation}
  Let $c_1,c_2,d_1,d_2 \in \outC$ and $x \in \outLp$ a primitive word.
  If for all $i,j \in \N$,
    $\distf(d_1^j c_1 [x^i], d_2^j c_2 [x^i]) \leq L$,
  then there exist $e_1,e_2 \in \outC$ such that
    for all $i,j \in \N$, $e_1 d_1^j c_1 [x^i] = e_2 d_2^j c_2 [x^i]$.
\end{lemma}

\begin{proof}
  Suppose that for all $i,j \in \N$,
    $\distf(d_1^j c_1 [x^i],d_2^j c_2 [x^i]) \leq L$.
  Then for all $i,j \in \N$,
    there exist $f_1,f_2 \in \Contexts{B}$,
    such that $f_1^{-1} d_1^j c_1 [x^i] = f_2^{-1} d_2^j c_2 [x^i]$
    and $|f_1| + |f_2| \leq L$.
  Let $C_L = \{ (f_1,f_2) \mid |f_1| + |f_2| \leq L \}$.
  $C_L$ is finite.
  Thus there exist some $(f_1,f_2) \in C_L$
    such that there is an infinite number of $i,j \in \N$
      such that $f_1^{-1} d_1^j c_1 [x^i] = f_2^{-1} d_2^j c_2 [x^i]$.
  Furthermore, there exists $i_0,j_0$ such that for all $i \geq i_0$, $j \geq j_0$,
    there exists $e_1,e_2 \in \Contexts{B}$
    and $e_1 d_1^{j-j_0} c_1 [x^{i-i_0}] = e_2 d_2^{j-j_0} c_2 [x^{i-i_0}]$.
  Finally, by applying \Cref{r:saarela} two times, we obtain that
    for all $i,j \in \N$,
      $e_1 d_1^j c_1 [x^i] = e_2 d_2^j c_2 [x^i]$.
\end{proof}


\begin{lemma}\label{r:comb-com-2-2-rev}
  Let $s_1, u_1, w_1, u_2, w_2, y_2 \in \outL$ and $t_1, x_2, v \in \outLp$
    such that $v$ is primitive, $|t_1| = |x_2|$
    and for all $i,j \in \N$,
      $s_1 t_1^j u_1 v^i w_1 = u_2 v^i w_2 x_2^j y_2$.
  Then there exist $f_1,f_2 \in \outC$
    such that for all $i,j \geq 1$
      there exists $k \in \N$
        such that $t_1^j u_1 v^i =  f_1 [v^k]$
        and $v^i w_2 x_2^j = f_2 [v^k]$.
\end{lemma}

\begin{proof}
  We can find sufficiently large $i_0$ and $j_0$ such that
    $t_1^{j_0}$ and $x_2^{j_0}$ both overlap with $v^{i_0}$
    with a common factor of length greater than
      $|t_1| + |v| - \gcd(|t_1|, |v|) = |x_2| + |v| - \gcd(|x_2|, |v|)$.
  Thus by \Cref{r:fw},
    $\proot(t_1) \sim \proot(x_2) \sim \proot(v)$.
  As $|t_1|=|x_2|$, $|t_1 u_1 v| \geq |x_2|$.
  Yet $t_1 u_1 v$ is a factor of $x_2^*$
    and, by \Cref{r:prootf-of-factor}, $\prootf(t_1 u_1 v) \sim \proot(x_2)$.
  Similarly, $\prootf(v w_2 x_2) \sim \proot(t_1)$.
  Thus $\proot(t_1) \sim \proot(x_2)
    \sim \prootf(t_1 u_1 v) \sim \prootf(v w_2 x_2)
    \sim \proot(v)$.

  Let $z_1,z_1',z_2,z_2' \in \outL$ and $\alpha \geq 1$
  such that $z_1 z_1' = z_2 z_2'$,
  and $v = z_1 z_1'$,
    $t_1 = (z_1' z_1)^\alpha$
    and $x_2 = (z_2' z_2)^\alpha$.
  Also we have
  $t_1 u_1 v = z_1' (z_1 z_1')^{\theta_1}$,
  $v w_2 x_2 = (z_1 z_1')^{\theta_2} z_2$,
  for some $\theta_1, \theta_2 \geq 0$.

  We choose $f_1 = (z_1', (z_1 z_1')^{\theta_1})$,
  and $f_2 = ((z_1 z_1')^{\theta_2}, z_2)$.
  Then, for all $i,j \geq 1$, we have
    $t_1^j u_1 v^i
      = z_1' (z_1 z_1')^{\alpha(j-1) + \theta_1 + (i-1)}$
  and
    $v^i w_2 x_2^j
      = (z_1 z_1')^{(i-1) + \theta_2 + \alpha(j-1)} z_2$.
  Let $k = (i-1) + \alpha(j-1) \geq 0$.
  And we obtain $t_1^j u_1 v^i = f_1 [ v^k ]$
  and $v^i w_2 x_2^j = f_2 [ v^k ]$.
\end{proof}

\begin{lemma}\label{r:comb-com-2-2}
  Let $s_1, u_1, w_1, s_2, u_2, w_2 \in \outL$ and $t_1, t_2, v \in \outLp$
    such that $|t_1| = |t_2|$
    and for all $i,j \in \N$,
      $s_1 t_1^j u_1 v^i w_1 = s_2 t_2^j u_2 v^i w_2$.
  Then there exist some contexts $f,g,h_1,h_2 \in \outC$
    such that for all $i,j \in \N$,
      $t_1^j u_1 v^i = h_1 g^j f [v^i]$
      and  $t_2^j u_2 v^i = h_2 g^j f [v^i]$.
\end{lemma}

\begin{proof}
  There exists $j_0$ sufficiently large such that
    $t_1^{j_0}$ overlap with $t_2^{j_0}$
    with a common factor of length greater than
      $|t_1| + |t_2| - \gcd(|t_1|, |t_2|) = |t_1| = |t_2|$.
  Thus by \Cref{r:fw},
    $\proot(t_1) \sim \proot(t_2)$.
  If $|s_1| \leq |s_2|$, let $t$ such that $s_1 t = s_2$ and $t_1 t = t t_2$.
  If $|s_1| \geq |s_2|$, let $t$ such that $s_1 = s_2 t$ and $t t_1 = t_2 t$.
  If $|w_1| \leq |w_2|$, let $i_0$ such that $v^{i_0} w_1 = w_2$ and $v^{i_0} \issuffix u_1$.
  If $|w_1| \geq |w_2|$, let $i_0$ such that $w_1 = v^{i_0} w_2$ and $v^{i_0} \issuffix u_2$.
  We obtain that
  \begin{itemize}
    \item $s_1 t = s_2$, $t_1 t = t t_2$, $u_1 = t u_2 v^{i_0}$, $v^{i_0} w_1 = w_2$, or
    \item $s_1 = s_2 t$, $t t_1 = t_2 t$, $t u_1 = u_2 v^{i_0}$, $v^{i_0} w_1 = w_2$, or
    \item $s_1 t = s_2$, $t_1 t = t t_2$, $u_1 v^{i_0} = t u_2$, $w_1 = v^{i_0} w_2$, or
    \item $s_1 = s_2 t$, $t t_1 = t_2 t$, $t u_1 v^{i_0} = u_2$, $w_1 = v^{i_0} w_2$.
  \end{itemize}

  We handle the first case. The others are similar.
  We choose
    $h_1 = (t, v^{i_0})$,
    $h_2 = (\eps, \eps)$,
    $g = (t_2, \eps)$,
    and $f = (u_2, \eps)$.
  Then for all $i \in \N$,
  $t_1^j u_1 v^i
    = t_1^j t u_2 v^{i_0} v^i
    = t t_2^j u_2 v^i v^{i_0}
    = h_1 g^j f [v^i]$
  and
  $t_2^j u_2 v^i
    = h_2 g^j f [v^i]$.
\end{proof}

\begin{lemma}\label{r:comb-com-2-3}
  Let $s_1, u_1, w_1, s_2, u_2, w_2, y_2 \in \outL$ and $t_1, t_2, x_2, v \in \outLp$
    such that $|t_1| = |t_2| + |x_2|$
    and for all $i,j \in \N$,
      $s_1 t_1^j u_1 v^i w_1 = s_2 t_2^j u_2 v^i w_2 x_2^j y_2$.
  Then there exist $f_1,f_2 \in \outC$
    such that for all $i,j \geq 1$
      there exists $k \in \N$
        such that $t_1^j u_1 v^i =  f_1 [v^k]$
        and $t_2^j u_2 v^i w_2 x_2^j = f_2 [v^k]$.
\end{lemma}

\begin{proof}
  We can find sufficiently large $i_0$ and $j_0$ such that
    $t_1^{j_0}$ overlap with $t_2^{j_0}$, $v^{i_0}$ and $x_2^{j_0}$
    with common factors of respective length greater than
      $|t_1| + |t_2| - \gcd(|t_1|, |t_2|)$,
      $|t_1| + |v| - \gcd(|t_1|, |v|)$, and
      $|t_1| + |x_2| - \gcd(|t_1|, |x_2|)$.
  Thus by \Cref{r:fw},
    $\proot(t_1) \sim \proot(t_2) \sim \proot(x_2) \sim \proot(v)$.
  As $|t_1|=|x_2|$, $|t_1 u_1 v| \geq |x_2|$.
  Similarly to \Cref{r:comb-com-2-2-rev}, we can show that
    $\prootf(t_1 u_1 v) \sim \prootf(t_2 u_2 v) \sim \prootf(v w_2 x_2) \sim \proot(v)$,
    and reconstruct the words to obtain the result.
\end{proof}

\begin{lemma}\label{r:comb-com-3-3-weak}
  Let $s_1, u_1, w_1, y_1, s_2, u_2, w_2, y_2 \in \outL$ and $t_1, x_1, t_2, x_2, v \in \outLp$
    such that $|t_1| + |t_2| = |x_1| + |x_2|$, $|t_1| \neq |x_1|$ and $|t_2| \neq |x_2|$,
    and for all $i,j \in \N$,
      $s_1 t_1^j u_1 v^i w_1 x_1^j y_1 = s_2 t_2^j u_2 v^i w_2 x_2^j y_2$.
  Then there exist $f_1,f_2 \in \outC$
    such that for all $i,j \geq 1$
      there exists $k \in \N$
        such that $t_1^j u_1 v^i w_1 x_1^j =  f_1 [v^k]$
        and $t_2^j u_2 v^i w_2 x_2^j = f_2 [v^k]$.
\end{lemma}

\begin{proof}
  Without loss of generality, consider that $|t_1| > |t_2|$.
  We can find sufficiently large $i_0$ and $j_0$ such that
    $t_1^{j_0}$ overlap with $t_2^{j_0}$, $v^{i_0}$ and $x_2^{j_0}$
    with common factors of respective length greater than
      $|t_1| + |t_2| - \gcd(|t_1|, |t_2|)$,
      $|t_1| + |v| - \gcd(|t_1|, |v|)$, and
      $|t_1| + |x_2| - \gcd(|t_1|, |x_2|)$.
  Thus by \Cref{r:fw},
    $\proot(t_1) \sim \proot(t_2) \sim \proot(x_2) \sim \proot(v)$.
  Similarly, we can show that $\proot(x_1) \sim \proot(x_2)$.
  Finally, similarly to \Cref{r:comb-com-2-2-rev}, we can show that
   $\prootf(t_1 u_1 v) \sim \prootf(t_2 u_2 v) \sim \prootf(v w_1 x_1) \sim \prootf(v w_2 x_2)
   \sim \proot(v)$,
   and reconstruct the words to obtain the result.
\end{proof}

\begin{lemma}\label{r:comb-com-3-3-strong}
  Let $s_1, u_1, w_1, y_1, s_2, u_2, w_2, y_2 \in \outL$ and $t_1, x_1, t_2, x_2, v \in \outLp$
    such that $|t_1| = |x_1|$, $|t_2| = |x_2|$,
    and for all $i,j \in \N$,
      $s_1 t_1^j u_1 v^i w_1 x_1^j y_1 = s_2 t_2^j u_2 v^i w_2 x_2^j y_2$.
  Then there exist some contexts $f,g,h_1,h_2 \in \outC$
    such that for all $i,j \in \N$,
      $t_1^j u_1 v^i w_1 x_1^j = h_1 g^j f [v^i]$
      and  $t_2^j u_2 v^i w_2 x_2^j = h_2 g^j f [v^i]$.
\end{lemma}

\begin{proof}
  There exists $j_0$ sufficiently large such that
    $t_1^{j_0}$ overlap with $t_2^{j_0}$
    with a common factor of length greater than
      $|t_1| + |t_2| - \gcd(|t_1|, |t_2|) = |t_1| = |t_2|$.
  Thus by \Cref{r:fw},
    $\proot(t_1) \sim \proot(t_2)$.
  The same applies to $x_1^{j_0}$ and $x_2^{j_0}$,
    and $\proot(x_1) \sim \proot(x_2)$.
  By a reasoning similar to \Cref{r:comb-com-2-2},
    based on the length of $s_1$, $s_2$, $y_1$ and $y_2$,
    we can show that there exists $t, x \in \outL$
    such that we obtain
  \begin{itemize}
    \item $s_1 t = s_2$, $t_1 t = t t_2$,
        $u_1 v^i w_1 = t u_2 v^i w_2 x$,
        $x x_1 = x_2 x$, $x y_1 = y_2$, or
    \item $s_1 t = s_2$, $t_1 t = t t_2$,
        $u_1 v^i w_1 x = t u_2 v^i w_2$,
        $x_1 x = x x_2$, $y_1 = x y_2$, or
    \item $s_1 = s_2 t$, $t t_1 = t_2 t$,
        $t u_1 v^i w_1 = u_2 v^i w_2 x$,
        $x x_1 = x_2 x$, $x y_1 = y_2$, or
    \item $s_1 = s_2 t$, $t t_1 = t_2 t$,
        $t u_1 v^i w_1 x = u_2 v^i w_2$,
        $x_1 x = x x_2$, $y_1 = x y_2$.
  \end{itemize}

  Again similarly to \Cref{r:comb-com-2-2}, we can reconstruct the words
  $t_1^j u_1 v^i w_1 x_1^j$ and $t_2^j u_2 v^i w_2 x_2^j$
  to obtain the result.
\end{proof}


\begin{proof}[Proof of \Cref{r:strongly-commuting-or-aligned}]
  Let $x \in \outLp$ a primitive word,
  $H_1 H_2$ be an $\Com{x}$ in $\T^k$,
  and let $\Delta = \split_c(x, H_1, H_2)$.
  Let $\rho_1:\; \ttrans{}{c_1} p_1 \ttrans{u_1}{d_1} q_1 \ttrans{u_2}{e_1} q_1$
  and $\rho_2:\; \ttrans{}{c_2} p_2 \ttrans{u_1}{d_2} q_2 \ttrans{u_2}{e_2} q_2$
    be two lassos in $\T_{\Delta}$.

  As $H_1 H_2$ is $\Com{x}$,
    we can find two synchronised $\Com{x}$ lassos in $\T$ around $p_1$ and $p_2$.
  Let $\rho_1':\; \ttrans{}{c_{0,1}} i_1 \ttrans{t_1}{c_{1,1}} p_1 \ttrans{t_2}{c_{2,1}} p_1$
  and $\rho_2':\; \ttrans{}{c_{0,2}} i_2 \ttrans{t_1}{c_{1,2}} p_2 \ttrans{t_2}{c_{2,2}} p_2$
    be those two lassos.
  Furthermore, by definition of $\split_c$, we have that
    $\Delta(p_1) = c_1$ and $\Delta(p_2) = c_2$, and
    for all $i \in \N$, there exists $k \in \N$
      such that $c_{2,1}^i c_{1,1} c_{0,1} [\eps] = c_1 [x^k]$
      and $c_{2,2}^i c_{1,2} c_{0,2} [\eps] = c_2 [x^k]$.
  By \Cref{d:ctp},
    for all $i,j \in \N$, $\distf(e_1^j d_1 c_1 [x^i], e_2^j d_2 c_2 [x^i]) \leq L$.

  By \Cref{r:com-dist-to-equation},
    there exist $f_1,f_2 \in \outC$ such that for all $i,j \in \N$,
      $f_1 e_1^j d_1 c_1 [x^i] = f_2 e_2^j d_2 c_2 [x^i]$.
  Then we have that $|e_1| = |e_2|$.
  We observe 10 cases.
  \medskip

  If $|e_1| = 0$ or $|e_2| = 0$ then $|e_1| = |e_2| = 0$ and $\rho_1,\rho_2$ are non-productive.
  \medskip

  If $e_1 \in \{\eps\} \times \outLp$ and $e_2 \in \outLp \times \{\eps\}$,
    then by \Cref{r:comb-com-2-2-rev},
    both $\rho_1$ and $\rho_2$ are productive and $\SCom{x}$.
  The same holds for the symmetrical case
    where $e_1 \in \outLp \times \{\eps\}$ and $e_2 \{\eps\} \times \outLp$.
  \medskip

  If $e_1, e_2 \in \outLp \times \{\eps\}$,
    then by \Cref{r:comb-com-2-2},
      there exist $f,g \in \outC$ such that
      both $\rho_1$ and $\rho_2$ are productive and $\SAlign{g,f,x}$.
  However,
    if it still happens that either one of $\rho_1$ and $\rho_2$ is $\SCom{x}$
      (implying that $\proot(\Lc{e_1}) \sim \proot(\Rc{e_1}) \sim x$
        or  $\proot(\Lc{e_2}) \sim \proot(\Rc{e_2}) \sim x$),
    then both $\rho_1$ and $\rho_2$ are $\SCom{x}$.
  If not, then they both are non-commuting.
  The same holds for the symmetrical case
    where $e_1, e_2 \in \{\eps\} \times \outLp$.
  \medskip

  If $e_1 \in \outLp \times \{\eps\}$ and $e_2 \in \outLp \times \outLp$,
    then by \Cref{r:comb-com-2-3},
    both $\rho_1$ and $\rho_2$ are productive and $\SCom{x}$.
  The same holds for the other three cases
    where exactly one of the four components of $e_1$ and $e_2$ is empty.
  \medskip

  If $e_1, e_2 \in \outLp \times \outLp$ and $\|e_1\| \neq \|e_2\|$,
    then by \Cref{r:comb-com-3-3-weak},
    both $\rho_1$ and $\rho_2$ are productive and $\SCom{x}$.
  \medskip

  If $e_1, e_2 \in \outLp \times \outLp$ and $\|e_1\| = \|e_2\|$,
    then by \Cref{r:comb-com-3-3-strong},
      there exist $f,g \in \outC$ such that
      both $\rho_1$ and $\rho_2$ are productive and $\SAlign{g,f,x}$.
  However,
    if it still happens that either one of $\rho_1$ and $\rho_2$ is $\SCom{x}$
      (implying that $\proot(\Lc{e_1}) \sim \proot(\Rc{e_1}) \sim x$
        or  $\proot(\Lc{e_2}) \sim \proot(\Rc{e_2}) \sim x$),
    then both $\rho_1$ and $\rho_2$ are $\SCom{x}$.
  If not, then they both are non-commuting.
\end{proof}


\begin{proof}[Proof of \Cref{r:all-strongly-commuting-or-aligned}]
  By \Cref{r:strongly-commuting-or-aligned}, if not productive,
    two lassos following a commuting lasso
    are either strongly commuting or strongly aligned.
  This result can be lifted to $k$ runs
    in a similar way to \Cref{r:all-commuting-or-aligned}.
\end{proof}

\subsection{Lassos Consecutive to a Non-Commuting Lasso}

We finally study the properties of lassos that are consecutive to a non-commuting lasso.
The following Lemma shows that only a non-commuting lasso can follow a non-commuting lasso.

\begin{lemma}\label{r:nc-nc}
  Let $f \in \outC$, $w \in \outL$,
    such that there exists no $x \in \outLp$ and $g \in \outC$
      such that, for all $i \in \N$, there exists $k \in \N$ such that $f^i [w] = g [x^k]$.
  Let $c, d \in \outC$, and $i \geq 1$.
  Then there exists no $x \in \outLp$ and $g \in \outC$
    such that, for all $j \in \N$, there exists $k \in \N$
      such that $d^j c f^i [w] = g [x^k]$.
\end{lemma}

\begin{proof}
  This can easily be proved by contradiction.
\end{proof}

The lassos after a non-commuting lasso are then always fully-aligned.
Given a \longStoC{} that satisfies the contextual twinning property,
  we can thus view its restriction after an aligned lasso
  as the pair of two classical finite state transducers
    that both satisfy the classical twinning property.




\begin{lemma}\label{r:split-nc-tp}
  Let $f \in \outC$ and $w \in \outL$
  and let $\Delta = \split_{nc}(f, w, H_1, H_2)$,
    for some $H_1 H_2$ a productive, non-commuting and $(f,w)-aligned$ lasso
      in $\T^{\leq |Q|}$.
  Then $\Rc{\T_\Delta}$ and $\Lc{\T_\Delta}$ both satisfy the $\TP{}$.
\end{lemma}

\begin{proof}
  We show the result for $\Rc{\T_\Delta}$. The proof for $\Lc{\T_\Delta}$ is symmetrical.

  Let $\ttrans{}{x} p \ttrans{u}{y} q \ttrans{v}{z} q$ and
  $\ttrans{}{x'} p' \ttrans{u}{y'} q' \ttrans{v}{z'} q'$
    two lassos in $\Rc{\T_\Delta}$.
  By the definition of $\Rc{\T_\Delta}$, there exist
    $(p,c),(p',c') \in \Delta$ such that $\Rc{c} = x$ and $\Rc{c'} = x'$, and
    $p\ttrans{u}{d} q \ttrans{v}{e} q$ and $p' \ttrans{u}{d'} q' \ttrans{v}{e'} q'$ in $\T_\Delta$
      such that $\Rc{d} = y$, $\Rc{d'} = y'$, $\Rc{e} = z$, $\Rc{e'} = z'$.

  From $\Delta = \split_{nc}(f, w, H_1, H_2)$ and $H_1 H_2$ being non-commuting,
  we know that there exist $i \geq 1$ and two lassos
    $\ttrans{}{c_1} o \ttrans{s}{c_2} p \ttrans{t}{c_3} p$ and
    $\ttrans{}{c_1'} o' \ttrans{s}{c_2'} p' \ttrans{t}{c_3'} p'$ in $\T$,
    such that $c_3 c_2 c_1 [\eps] = c f^i [w]$ and $c_3' c_2' c_1' [\eps] = c' f^i [w]$.
  Thus we can build two lassos in $\T$:
    $\rho_1:\ \ttrans{}{c_1} o \ttrans{s}{d c_3 c_2} q \ttrans{v}{e} q$ and
    $\rho_2:\ \ttrans{}{c_1'} o' \ttrans{s}{d' c_3' c_2'} q' \ttrans{v}{e'} q'$.

  By \Cref{d:ctp},
    for all $j \in \N$, $\distf(e^j d c f^i [w], e'^j d' c' f^i [w]) \leq L$.
  As $H_1 H_2$ is non-commuting,
    by \Cref{r:nc-nc}, $\rho_1$ and $\rho_2$ must also be non-commuting,
      and thus strongly-balanced.
  As $H_1 H_2$ and $\rho_1$ and $\rho_2$ are strongly-balanced and non-commuting
    the $f^i[w]$ part can only overlap with itself in the words
      $e^j d c f^i [w]$ and $e'^j d' c' f^i [w]$.
  Therefore, we can derive that for all $j \in \N$, $\dist(x y z^j, x' y' z'^j) \leq L$.
  %
  %
  %
  %
  %
  %
\end{proof}


\begin{lemma}\label{r:extract-nc-tp}
  Let $x \in \outLp$ a primitive word
  and let $\Delta_0 = \split_c(x, H_1, H_2)$,
    for some $H_1 H_2$ an productive and $\Com{x}$ lasso
      in $\T^{\leq |Q|}$.
  Let $g,f \in \outC$
  and let $\Delta = \extract_{nc}(g, f, x, \Delta_0, H_3, H_4)$,
    for some $H_3 H_4$ a productive, non-commuting and \linebreak$\SAlign{g,f,x}$ lasso
      in $\T_{\Delta_0}^{\leq |Q|}$.
  Then $\Rc{\T_\Delta}$ and $\Lc{\T_\Delta}$ both satisfy the $\TP{}$.
\end{lemma}

\begin{proof}
  We show the result for $\Rc{\T_\Delta}$. The proof for $\Lc{\T_\Delta}$ is symmetrical.

  Let $\ttrans{}{v_0} q_2 \ttrans{u_1}{v_1} q_3 \ttrans{u_2}{v_2} q_3$ and
  $\ttrans{}{v_0'} q_2' \ttrans{u_1}{v_1'} q_3' \ttrans{u_2}{v_2'} q_3'$
    two lassos in $\Rc{\T_\Delta}$.
  By the definition of $\Rc{\T_\Delta}$, there exist
    $(q_2,e_0),(q_2',e_0') \in \Delta$ such that $\Rc{e_0} = v_0$ and $\Rc{e_0'} = v_0'$, and
    $q_2 \ttrans{u_1}{e_1} q_3 \ttrans{u_2}{e_2} q_3$ and
    $q_2' \ttrans{u_1}{e_1'} q_3' \ttrans{u_2}{e_2'} q_3'$
      in $\T_\Delta$
      such that $\Rc{e_1} = v_1$, $\Rc{e_1'} = v_1'$, $\Rc{e_2} = v_2$, $\Rc{e_2'} = v_2'$.

  From $\Delta = \extract_{nc}(g, f, x, \Delta_0, H_3, H_4)$ and $H_3 H_4$ being $\SAlign{g,f}$,
  we know that there exist $i \geq 1$ and two lassos
    $\ttrans{}{d_0} q_1 \ttrans{t_1}{d_1} q_2 \ttrans{t_2}{d_2} q_2$ and
    $\ttrans{}{d_0'} q_1' \ttrans{t_1}{d_1'} q_2' \ttrans{t_2}{d_2'} q_2'$
    in $\T_{\Delta_0}$,
    such that $d_2 d_1 d_0 = e_0 g^i f$ and $d_2' d_1' d_0' = e_0' g^i f$.
  Thus we can build two lassos in $\T_{\Delta_0}$:
    $\ttrans{}{d_0} q_1 \ttrans{t_1 t_2 u_1}{e_1 d_2 d_1} q_3 \ttrans{u_2}{e_2} q_3$ and
    $\ttrans{}{d_0'} q_1' \ttrans{t_1 t_2 u_1}{e_1' d_2' d_1'} q_3' \ttrans{u_2}{e_2'} q_3'$.

  From $\Delta_0 = \split_{c}(x, H_1, H_2)$ and $H_1 H_2$ being $\Com{x}$,
  we know that there exist $k \geq 1$, $c \in \outC$ and two lassos
    $\ttrans{}{c_0} q_0 \ttrans{s_1}{c_1} q_1 \ttrans{s_2}{c_2} q_1$ and
    $\ttrans{}{c_0'} q_0' \ttrans{s_1}{c_1'} q_1' \ttrans{s_2}{c_2'} q_1'$
    in $\T$,
    such that $c_2 c_1 c_0 [\eps] = d_0 c [x^k]$ and $c_2' c_1' c_0' [\eps] = d_0' c [w^k]$.
  Thus we can build two lassos in $\T$:
    $\rho_1:\ \ttrans{}{c_0}
      q_0 \ttrans{s_1 s_2 t_1 t_2 u_1}{e_1 d_2 d_1 c_2 c_1}
      q_3 \ttrans{u_2}{e_2} q_3$ and
    $\rho_2:\ \ttrans{}{c_0'}
      q_0' \ttrans{s_1 s_2 t_1 t_2 u_1}{e_1' d_2' d_1' c_2' c_1'}
      q_3' \ttrans{u_2}{e_2'} q_3'$.

  By \Cref{d:ctp},
    for all $j \in \N$,
      $\distf(e_2^j e_1 e_0 g^i f c [x^k], e_2'^j e_1' e_0' g^i f c [x^k]) \leq L$.
  %
  %
  As $H_1 H_2$ is non-commuting,
    by \Cref{r:nc-nc}, $\rho_1$ and $\rho_2$ must also be non-commuting,
      and thus strongly-balanced.
  As $H_1 H_2$ and $\rho_1$ and $\rho_2$ are strongly-balanced and non-commuting
    the $g^i f c [x^k]$ part can only overlap with itself in the words
      $e_2^j e_1 e_0 g^i f c [x^k]$ and $e_2'^j e_1' e_0' g^i f c [x^k]$.
  Therefore, we can derive that for all $j \in \N$, $\dist(v_0 v_1 v_2^j, v_0' v_1' v_2'^j) \leq L$.
\end{proof}


\begin{proof}[Proof of \Cref{r:all-fully-aligned}]
  %
  By \Cref{r:split-nc-tp,r:extract-nc-tp},
    $\Rc{\T_\Delta}$ and $\Lc{\T_\Delta}$ both satisfy the $\TP{}$.
\end{proof}



%% file: 160-annex-construction.tex

\subsection{Observation on the $\Choose$ operator}

We have considered a functional transducer $\T$, and have assumed
that $\T$ is trim. As a consequence,
if we consider two runs 
$\ttrans{}{c_1} p_1 \ttrans{u}{d_1} q$
and
$\ttrans{}{c_2} p_2 \ttrans{u}{d_2} q$, then there exists
a run 
$q \ttrans{u}{e} f \ttrans{}{g}$ where $f$ is final.
By functionality, we have 
$ ged_1c_1[\eps]  = ged_2c_2[\eps]$,
hence 
$ d_1c_1[\eps] = d_2c_2[\eps] $. This implies that even if the $\Choose$
operator may select different contexts corresponding to different runs
leading to the same state, when they are applied to $\eps$,
they yield the same word.

Similarly, if we consider two runs $p \ttrans{u}{d_1} q$
and
$p \ttrans{u}{d_2} q$, and a word $w$ such that 
$\ttrans{}{c} i \ttrans{v}{d} q$ with $w = dc[\eps]$, then
we have $d_1[w]=d_2[w]$. As a consequence, the
choice realised by $\Choose$ has no impact as soon
as one compares the contexts applied to a possible output
word produced before.

Last, using a similar reasoning, we can prove that
when considering two runs $p_1 \ttrans{u}{d_1} q$
and
$p_2 \ttrans{u}{d_2} q$ such that $p_1$ and $p_2$ appear in 
some $\Delta:Q\pto \outC$, obtained after a non-commuting lasso,
then we have $d_1\Delta(p_1) = d_2\Delta(p_2)$. Hence, the choice
realised by $\Choose$ has actually no impact.

In the sequel, we will often write equalities involving partial functions
$\Delta:Q\pto \outC$. When these equalities are in one of the three above situations,
we will thus omit the operator $\Choose$, for simplicity of the writing.


\subsection{Additional Notations}

For a run H in $\T^k$, we denote by $\word(H)$ the word read by $H$.

Given a state $\pB = (x,\Delta,H) \in \Qb$, we let $x_{\pB}$ be equal to $x$.

Given a state $\pB = (x,\Delta,H) \in \Qb$
  and some run $H'$ in $\T^k$
  such that the start state of $H'$ is the end state of $H$,
  we let $\pB \bullet H' = (x,\Delta, H H')$.

When necessary, we extend the run notation to include the transducer name.
For instance, $\pB \ttrans[\D]{u}{c} \qB$ depicts
  a transition in $\D$ from state $\pB$ to state $\qB$
    reading the word $u$ and outputting the context $c$.

\subsection{Correctness}

We start with the following easy observation:
\begin{lemma}\label{r:act-non-prod}
  Let $\Delta : Q \pto \outC$
    and $H$ a history in $\T_\Delta^{\leq |Q|}$.
  If $H = H_1 H_2$, with $H_1 H_2$ a non-productive lasso.
    then $\Delta \act H_1 = (\Delta \act H_1) \act H_2$.
\end{lemma}



\begin{proof}[Proof of~\Cref{r:corr-extend}]
  We proceed to a case analysis.

  \case{$x_{\pB} = \eps$}
  Let $\pB = (\eps, \tinit, H_1)$.
  Either $H_2$ is non-productive, $\qB = \pB$, $c = \ceps$
    and, by \Cref{r:act-non-prod},
    $\Delta_{\pB} \act H_2 [\eps]
      = \tinit \act H_1 H_2 [\eps]
      = \tinit \act H_1 [\eps]
      = \Delta_{\pB} [\eps]
      = \Delta_{\qB} c [\eps]$.

  Otherwise, $H_2$ is productive
    and we pass through the $\textbf{else if}$ block at Line~\ref{extend-s}.

  If $H_1 H_2$ is $\Com{x}$ for some $x \in \outLp$,
    then $\Delta = \split_c(x, H_1, H_2)$,
    $k = \pow_c(x, H_1, H_2)$,
    $\qB = (x, \Delta, id_\Delta)$
    and $c = (\eps, x^k)$.
  Thus, $\Delta_{\pB} \act H_2 [\eps]
    = \tinit \act H_1 H_2 [\eps]
    = \Delta [x^k]
    = \Delta_{\qB} c [\eps]$.

  If $H_1 H_2$ is $\Align{f,w}$ for some $f \in \outC$ and $w \in \outL$,
    then $\Delta = \split_{nc}(f, w, H_1, H_2)$,
    $\qB = (\bot, \Delta, id_\Delta)$
    and $c = f \cdot (\eps, w)$.
  Thus, $\Delta_{\pB} \act H_2 [\eps]
    = \tinit \act H_1 H_2 [\eps]
    = \Delta f [w]
    = \Delta_{\qB} c [\eps]$.

  \case{$x_{\pB} \in \outLp$}
  Let $\pB = (x, \Delta_0, H_1)$ and $j \in \N$.
  Either $H_2$ is non-productive, $\qB = \pB$, $c = \ceps$
    and, by \Cref{r:act-non-prod},
    $\Delta_{\pB} \act H_2 [x^j]
      = \Delta_0 \act H_1 H_2 [x^j]
      = \Delta_0 \act H_1 [x^j]
      = \Delta_{\pB} [x^j]
      = \Delta_{\qB} c [x^j]$.

  Otherwise, $H_2$ is productive
    and we pass through the $\textbf{else if}$ block at Line~\ref{extend-c}.

  If $H_1 H_2$ is $\SCom{x}$,
    then $k = |\out(H_2)|/|x|$,
    $\qB = \pB$
    and $c = (\eps, x^k)$.
  Thus, $\Delta_{\pB} \act H_2 [x^j]
    = \Delta_0 \act H_1 H_2 [x^j]
    = \Delta_0 \act H_1 [x^{j+k}]
    = \Delta_{\qB} c [x^j]$.

  If $H_1 H_2$ is $\SAlign{g,f,x}$ for some $g,f \in \outC$,
    then $\Delta = \linebreak\extract_{nc}(g, f, x, \Delta_0, H_1, H_2)$,
    $\qB = (\bot, \Delta, id_\Delta)$
    and $c = g f$.
  Thus, $\Delta_{\pB} \act H_2 [x^j]
    = \Delta_0 \act H_1 H_2 [x^j]
    = \Delta g f [x^j]
    = \Delta_{\qB} c [x^j]$.
\end{proof}


\begin{proof}[Proof of~\Cref{r:corr-simplify}]
The result  follows from~\Cref{r:corr-simplify-nc,r:corr-simplify-s&c}.
\end{proof}

\begin{lemma}[Correctness of $\Call{simplify}$ for non-commuting states]\label{r:corr-simplify-nc}
  Let $\pB \in \Qbi$
    such that $x_{\pB} = \bot$
    and $(\qB, c) = \Call{simplify}{\pB}$.
  Then $\qB \in \Qb$ and $\Delta_{\pB} = \Delta_{\qB} c$.
\end{lemma}

\begin{proof}
  Let $\pB = (\bot, \Delta, H)$.
  The fact that $\qB \in \Qb$ is trivial.
  As $x_{\pB} = \bot$,
    we only pass through the $\textbf{if}$ block at Line~\ref{simplify-nc}.
  Let $\Delta' = \Delta \act H$ and $c = \lcc(\Delta')$.
    and $\qB = (\bot, \Delta' \cdot c^{-1}, id_{\Delta'})$.
  Thus $\Delta_{\qB} c = (\Delta \act H) \cdot c^{-1} c = \Delta_{\pB}$.
\end{proof}


\begin{lemma}[Correctness of $\Call{simplify}$ for startup and commuting states]\label{r:corr-simplify-s&c}
  Let $\pB \in \Qbi$
    such that $x_{\pB} \neq \bot$
    and $(\qB, c) = \Call{simplify}{\pB}$.
    Then $\qB\in\Qb$ and we have:
  \begin{itemize}
    \item If $x_{\pB} = \eps$
      then $\Delta_{\pB} [\eps] = \Delta_{\qB} c [\eps]$.
    \item If $x_{\pB} \in \outLp$
      then for all $k \in \N$, $\Delta_{\pB} [x_{\pB}^k] = \Delta_{\qB} c [x_{\pB}^k]$.
  \end{itemize}
\end{lemma}

\begin{proof}
  First observe that the fact that $\qB \in \Qb$ can be proven using
  a simple induction.

  We now consider the second property, and proceed by strong induction on $|H_{\pB}|$.
  If $H_{\pB} = id$ then, as $x_{\pB} \neq \bot$,
    we only pass through the $\textbf{else}$ statement at Line~\ref{simplify-last-else}.
  Then the result is trivially obtained.

  \medskip
  Otherwise $|H_{\pB}| > 0$.
  If $H_{\pB}$ doesn't contain a loop then, again,
    we only pass through the $\textbf{else}$ statement at Line~\ref{simplify-last-else},
    and the result is trivially obtained.

  \medskip
  If $H_{\pB}$ contains a loop,
    we pass through the $\textbf{else if}$ block at Line~\ref{simplify-loop}.
  Let $\pB = (x, \Delta, H)$ and $H = H_1 H_2 H_3$
    where $H_1 H_2$ is the first lasso in $H_{\pB}$.
  Let $\qB = (x, \Delta, H_1)$,
    $(\rB, c) = \Call{extend\_with\_loop}{\qB, H_2}$,
    and $(\sB, d) = \Call{simplify}{\rB \act H_3}$.
  We observe two cases.

  \case{$x_{\pB} = \eps$}
  By \Cref{r:corr-extend}, we have that
    $\tinit \act H_1 \act H_2 [\eps] = \Delta_{\qB} \act H_2 [\eps] = \Delta_{\rB} c [\eps]$.

  If $x_{\rB} = \bot$
    then, by \Cref{r:corr-simplify-nc}, $\Delta_{\rB} \act H_3 = \Delta_{\sB} d$.
  We obtain that
  \begin{align*}
    \Delta_{\pB} [\eps] &= \tinit \act H_1 H_2 H_3 [\eps] \\
      &= (\tinit \act H_1 H_2) \act H_3 [\eps] \\
      &= (\Delta_{\rB} c) \act H_3 [\eps] \\
      &= (\Delta_{\rB} \act H_3) c [\eps] \\
      &= \Delta_{\sB} d c [\eps]
  \end{align*}

  If $x_{\rB} \neq \bot$,
    as $|H_{\rB \act H_3}| \leq |H_1| + |H_3| < |H_{\pB}|$,
    then, by the induction hypothesis,
    \Cref{r:corr-simplify-s&c} holds for $\Call{simplify}{\rB \act H_3}$.

  If $x_{\rB} = \eps$,
    we have that $c = \ceps$
    and $\Delta_{\rB} \act H_3 [\eps] = \Delta_{\sB} d [\eps]$.
  We obtain that
  \begin{align*}
    \Delta_{\pB} [\eps] &= \tinit \act H_1 H_2 H_3 [\eps] \\
      &= (\tinit \act H_1 H_2) \act H_3 [\eps] \\
      &= (\Delta_{\rB} c) \act H_3 [\eps] \\
      &= \Delta_{\rB} \act H_3 [\eps] \\
      &= \Delta_{\sB} d [\eps] \\
      &= \Delta_{\sB} d c [\eps]
  \end{align*}

  If $x_{\rB} \in \outLp$,
    we have that $c = (\eps, x^\ell)$ for some $\ell \in \N$
    and for all $k \in \N$, $\Delta_{\rB} \act H_3 [x^k] = \Delta_{\sB} d [x^k]$.
  We obtain that
  \begin{align*}
    \Delta_{\pB} [\eps] &= \tinit \act H_1 H_2 H_3 [\eps] \\
      &= (\tinit \act H_1 H_2) \act H_3 [\eps] \\
      &= (\Delta_{\rB} c) \act H_3 [\eps] \\
      &= \Delta_{\rB} \act H_3 [x^\ell] \\
      &= \Delta_{\sB} d [x^\ell] \\
      &= \Delta_{\sB} d c [\eps]
  \end{align*}

  \case{$x_{\pB} \in \outLp$}
  By \Cref{r:corr-extend}, we have that
    for all $k \in \N$,
    $\tinit \act H_1 \act H_2 [x^k] = \Delta_{\qB} \act H_2 [x^k] = \Delta_{\rB} c [x^k]$.
  Let $j \in \N$.

  If $x_{\rB} = \bot$,
    then, by \Cref{r:corr-simplify-nc}, $\Delta_{\rB} \act H_3 = \Delta_{\sB} d$.
  We obtain that
  \begin{align*}
    \Delta_{\pB} [x^j] &= \tinit \act H_1 H_2 H_3 [x^j] \\
      &= (\tinit \act H_1 H_2) \act H_3 [x^j] \\
      &= (\Delta_{\rB} c) \act H_3 [x^j] \\
      &= (\Delta_{\rB} \act H_3) c [x^j] \\
      &= \Delta_{\sB} d c [x^j]
  \end{align*}

  If $x_{\rB} \neq \bot$,
    as  $|H_{\rB \act H_3}| = |H_1| + |H_3| < |H_{\pB}|$,
    then, by the induction hypothesis,
    \Cref{r:corr-simplify-s&c} holds for $\Call{simplify}{\rB \act H_3}$.

  Also, as $x_{\qB} = x_{\pB} \in \outLp$,
    by construction, it can only happen that $x_{\rB} \in \outLp$.
  Thus we have that $c = (\eps, x^\ell)$ for some $\ell \in \N$
    and for all $k \in \N$, $\Delta_{\rB} \act H_3 [x^k] = \Delta_{\sB} d [x^k]$.
  We obtain that
  \begin{align*}
    \Delta_{\pB} [x^j] &= \tinit \act H_1 H_2 H_3 [x^j] \\
      &= (\tinit \act H_1 H_2) \act H_3 [x^j] \\
      &= (\Delta_{\rB} c) \act H_3 [x^j] \\
      &= \Delta_{\rB} \act H_3 [x^{j+\ell}] \\
      &= \Delta_{\sB} d [x^{j+\ell}] \\
      &= \Delta_{\sB} d c [x^j] \qedhere
  \end{align*}
\end{proof}


\begin{lemma}[Correctness of transitions]\label{r:corr-trans}
  For all $\qB \in \Qb$ such that $\iB \ttrans[\D]{u}{c} \qB$,
    $\Delta_{\qB} c [\eps] = \tinit \act u [\eps]$.
\end{lemma}

\begin{proof}
  We proceed by induction on $|u|$.
  If $u = \eps$, the result is obtained trivially.
  If $u = u' a$ with $a \in \inA$,
    let $\pB = (x, \Delta, H) \in \Qb$
    such that $\iB \ttrans{u'}{c} \pB \ttrans{a}{d} \qB$,
    and $(\qB, d) = \Call{simplify}{(x, \Delta, H \act a)}$.
  By the hypothesis of induction, we obtain that $\tinit \act u' [\eps] = \Delta_{\pB} c [\eps]$.
  By extending with $a$, we get that
    $\tinit \act u [\eps] = \tinit \act u' \act a [\eps]
      = \Delta_{\pB} c \act a [\eps]
      = (\Delta_{\pB} \act a) c [\eps]$.
  We observe three cases.
  If $\pB \in \Qbs$ then $c = \ceps$
    and, by \Cref{r:corr-simplify-s&c},
    $(\Delta_{\pB} \act a) c [\eps]
      = \Delta_{\pB} \act a [\eps]
      = \Delta_{\qB} d [\eps]
      = \Delta_{\qB} d c [\eps]$.
  If $\pB \in \Qbc$ then $c [\eps] = x^k$, for some $k \in \N$,
    and, by \Cref{r:corr-simplify-s&c},
    $(\Delta_{\pB} \act a) c [\eps]
      = \Delta_{\pB} \act a [x^k]
      = \Delta_{\qB} d [x^k]
      = \Delta_{\qB} d c [\eps]$.
  Finally, if $\pB \in \Qbnc$ then,
    by \Cref{r:corr-simplify-nc},
    $(\Delta_{\pB} \act a) c [\eps]
      = \Delta_{\qB} d c [\eps]$.
\end{proof}



\subsection{Boundedness}

\begin{lemma}\label{r:ewl-shorter-history}
  Let $\pB = (x, \Delta, H_1)$ and $\qB = (x', \Delta', H')$,
    $c \in \outC$ and $H_2 \in \Hc(\T_{\Delta \act H_1}^{\leq |Q|})$,
    such that $(\qB, c) = \Call{extend\_with\_loop}{\pB, H_2}$.
  Then $|word(H')| \leq |word(H_1)|$.
\end{lemma}

\begin{proof}
  We obtain the result by a trivial case analysis of \Call{extend\_with\_loop}{}.
\end{proof}


\begin{lemma}
  Let $\pB = (x, \Delta, H)$ and $\sB = (x', \Delta', H')$,
    and $c \in \outC$
    such that $(\sB, c) = \Call{simplify}{\pB}$.
  Then $|word(H')| < \QQ$.
\end{lemma}

\begin{proof}
  We proceed by strong induction on the length of $word(H)$.

  If $H = id$ then
    we can only pass through the $\textbf{if}$ block at Line~\ref{simplify-nc}
    or the $\textbf{else if}$ block at Line~\ref{simplify-last-else}.
  In both cases, the result is obtained trivially.

  Otherwise, let $n = |word(H)|$, and we observe three cases.
  Firstly, if $x = \bot$,
    then we pass through the $\textbf{else if}$ block at Line~\ref{simplify-last-else},
    and the result is obtained trivially.
  Secondly, if there is a loop in $H$,
    then we pass through the $\textbf{else if}$ block at Line~\ref{simplify-loop}.
    Let $H = H_1 H_2 H_3$ where $H_1 H_2$ is the first lasso in $H_{\pB}$.
    Let $\qB = (x, \Delta, H_1)$,
      $(\rB, c) = \Call{extend\_with\_loop}{\qB, H_2}$,
      and $(\sB, d) = \Call{simplify}{\rB \act H_3}$.
    By \Cref{r:ewl-shorter-history}, we have that $|word(H_{\rB})| \leq |word(H_1)|$.
    Thus $|word(H_{\rB \act H_3})| = |word(H_{\rB} H_3)| \leq |word(H_1 H_3)| < |word(H)|$,
      and, by the hypothesis of induction applied on $\rB \act H_3$, we obtain the result.
  Thirdly, if there is no loop in $H$,
    then we pass through the $\textbf{else if}$ block at Line~\ref{simplify-last-else}
    and we have $|word(H')| = |word(H)| < \QQ$.
    Indeed, suppose we had that $|word(H)| \geq \QQ$,
      then $H$ must contain a loop, which is a contradiction.
\end{proof}


\begin{lemma}\label{r:ewl-small-x-and-delta}
  Let $\pB = (x, \Delta, H_1)$ and $\qB = (x', \Delta', H')$,
    $c \in \outC$ and $H_2 \in \Hc(\T_{\Delta \act H_1}^{\leq |Q|})$,
    such that $(\qB, c) = \Call{extend\_with\_loop}{\pB, H_2}$.
  We assume that $|word(H_1 H_2)| \leq \QQ$,
    $|x| \leq \base \QQ$
    and for all $(q,d) \in \Delta$, $|d| \leq \base \QQ$.
  Then we distinguish two cases:
  \begin{itemize}
  \item if $x'\neq \bot$, then  $|x'| \leq \base \QQ$
    and for all $(q,d) \in \Delta'$, $|d| \leq \base \QQ$,
    \item if $x'=\bot$, then for all $(q,d) \in \Delta'$, $|d| \leq 2 \base \QQ$.
    \end{itemize}
\end{lemma}

\begin{proof}
  We proceed by case analysis of \Call{extend\_with\_loop}{}.
  If $H_2$ is non-productive
    then we pass through the $\textbf{if}$ block at Line~\ref{extend-np},
    the returned state is $\pB$
    and the result is obtained trivially.
  Otherwise, $H_2$ is productive.

  If $x_{\pB} = \eps$
    then we pass through the $\textbf{else if}$ block at Line~\ref{extend-s}.
  Either $H_1 H_2$ is $\Com{x}$, for some $x \in \outLp$,
    and we let $\Delta' = \split_c(x, H_1, H_2)$,
  or $H_1 H_2$ is $\Align{f,w}$, for some $f \in \outC$ and $w \in \outL$,
    and we let $\Delta' = \split_{nc}(f, w, H_1, H_2)$.
  In both cases, by definition of $\Delta'$, we have that,
    for all $(q,d) \in \Delta'$, $|d| \leq \base + |\out(H_1)| \leq \base \QQ$,
    because $|\word(H_2)| \geq 1$ and $|word(H_1 H_2)| \leq \QQ$.

  If $x_{\pB} \in \outLp$
    then we pass through the $\textbf{else if}$ block at Line~\ref{extend-s}.
  If $H_1 H_2$ is \linebreak$\SCom{x}$,
    then the returned state is $\pB$
    and the result is obtained trivially.
  If $H_1 H_2$ is $\SAlign{g,f,x}$, for some $f,g \in \outC$,
    then we let $\Delta' = \extract_{nc}(g, f, x, \Delta, H_1, H_2)$.
    By definition of $\Delta'$, we have that,
      for all $(q,d) \in \Delta'$, $|d| \leq |\Delta(q)| + |\out(H_1)| \leq 2 \base \QQ$,
      because $ |\Delta(q)| \le \base \QQ$ and $|word(H_1)| \leq \QQ$.
\end{proof}


\begin{lemma}\label{r:simplify-small-x-and-delta}
  Let $\pB = (x, \Delta, H)$ and $\sB = (x', \Delta', H')$,
    and $c \in \outC$
    such that $(\sB, c) = \Call{simplify}{\pB}$
    and $x,x'\in \outL$.
  If $|word(H)| \leq \QQ$,
    $|x| \leq \base \QQ$
    and for all $(q,d) \in \Delta$, $|d| \leq \base \QQ$
  then $|x'| \leq \base \QQ$
    and for all $(q,d) \in \Delta'$, $|d| \leq \base \QQ$.
\end{lemma}

\begin{proof}
  We proceed by strong induction on the length of $word(H)$.
%
%
%
  Let $n = |word(H)|$, we observe two cases.
  %
  %
  First, if there is a loop in $H$,
    then we pass through the $\textbf{else if}$ block at Line~\ref{simplify-loop}.
    Let $H = H_1 H_2 H_3$ where $H_1 H_2$ is the first lasso in $H_{\pB}$.
    Let $\qB = (x, \Delta, H_1)$,
      $(\rB, c) = \Call{extend\_with\_loop}{\qB, H_2}$,
      $\rB = (x'', \Delta'', H'')$,
      and $(\sB, d) = \Call{simplify}{\rB \act H_3}$.
    We have that $|word(H_1 H_2)| \leq |word(H)| \leq \QQ$.
    Then, by \Cref{r:ewl-small-x-and-delta}, we have that
      $|x''| \leq \base \QQ$ and for all $(q,d) \in \Delta''$, $|d| \leq \base \QQ$.
    Again, by \Cref{r:ewl-shorter-history}, we have that $|word(H_{\rB})| \leq |word(H_1)|$.
    Thus $|word(H_{\rB \act H_3})| = |word(H_{\rB} H_3)| \leq |word(H_1 H_3)| < |word(H)| \leq \QQ$,
      and, by the hypothesis of induction applied on $\rB \act H_3$, we obtain the result.
  Second, if there is no loop in $H$,
    then we pass through the $\textbf{else if}$ block at Line~\ref{simplify-last-else}
    and the result is obtained trivially.
\end{proof}


\begin{lemma}[Boundedness of $\D$]\label{r:boundedness}
  For all $\qB = (x', \Delta', H') \in \Qb$
    such that $\iB \ttrans[D]{u}{c} \qB$,
    the following assertions are satisfied:
  \begin{itemize}
    \item $|x'| \leq \base \QQ$,
    \item $|word(H')| < \QQ$,
    \item if $x'\neq \bot$, then for all $(q,d) \in \Delta'$, $|d| \leq \base \QQ$,
    \item if $x'=\bot$, then for all $(q,d) \in \Delta'$, $|d| \leq 4 \base |Q|^{|Q|+2} $.
  \end{itemize}
\end{lemma}

\begin{proof}
  We distinguish two cases, whether $x'=\bot$ or not.

  We start with the case $x'\neq \bot$ and proceed by induction on $|u|$.
  If $u = \eps$, the result is obtained trivially.
  If $u = u' a$ with $a \in \inA$,
    let $\pB = (x, \Delta, H) \in \Qb$
    such that $\iB \ttrans{u'}{c} \pB \ttrans{a}{d} \qB$,
    and $(\qB, d) = \Call{simplify}{(x, \Delta, H \act a)}$.
  By the hypothesis of induction, we have that
    $|x| \leq \base \QQ$,
    for all $(q,d) \in \Delta$, $|d| \leq \base \QQ$,
    and $|word(H)| < \QQ$.
  By extending with $a$, we have that
    $|word(H \act a)| \leq \QQ$.
  Then by \Cref{r:simplify-small-x-and-delta},
    we obtain the result.

   We now consider that $x'=\bot$. The execution $\iB \ttrans[D]{u}{c} \qB$
   can be decomposed as 
   $\iB \ttrans[D]{u_1}{c_1} \pB \ttrans[D]{a}{c_2} \pB' \ttrans[D]{u_2}{c_3} \qB$,
   with $x_{\pB} \neq \bot$ and $x_{\pB'} = \bot$.	
   The transition from $\pB$ to $\pB'$ involves the removal of a loop,
   by $\Call{extend\_with\_loop}{}$, which is non-commuting.
   As a consequence of the \twoloop{}, we have that some intermediate state 
   $\pB'' = (\bot,\Delta,id_\Delta)$
   is computed, and that $\Lc{\T_\Delta}$ and $\Rc{\T_\Delta}$ both satisfy the
   classical twinning property. In addition, thanks to the first case of this proof,
   and to~\Cref{r:ewl-small-x-and-delta},
   we also have that for all $(q,d) \in \Delta'$, $|d| \leq 2 \base \QQ$.
   The behaviour of our procedure starting from this intermediate state $\pB''$
   is exactly the one of the
   determinisation procedure of Choffrut (performed on the two sides of the context).
   See for details the Line~24 of~\Cref{simplify}.
   Thanks to results of~\cite{Choffrut77,BealC02},
   we know that delays stored in the determinisation construction of Choffrut 
   have size at most $2n^2M$, where $M$ is the size of the largest output of the transducer,
   and $n$ is the number of states. As a consequence, we obtain that for all  $(q,d) \in \Delta'$, 
   we have :
   $$|d| \leq 2 (|Q|^2 ). (2 \base \QQ) = 4 \base |Q|^{|Q|+2} $$
\end{proof}

\subsection{Proof of~\Cref{t:determinisation}}

\begin{proof}[Proof of~\Cref{t:determinisation}]
The fact that $\D$ is deterministic is direct by an observation of its definition.

To prove the equivalence between $\D$ and $\T$, we consider a word $u\in \inL$, and
the run $\iB \ttrans[\D]{u}{c} \qB$ of $\D$ on $u$.
By~\Cref{r:corr-trans}, we have $\Delta_{\qB} c [\eps] = \tinit \act u [\eps]$.
This entails:
$$
\begin{array}{ccl}
u\in\dom(\inter{\T}) & \iff & \dom(\tinit \act u)\cap\dom(\tfinal) \neq \emptyset \\
& \iff & \dom(\Delta_{\qB})\cap\dom(\tfinal) \neq \emptyset \\
& \iff & \qB \in\dom(\overline{\tfinal})\\
& \iff &  u\in\dom(\inter{\D})
\end{array}
$$
The definition of $\overline{\tfinal}$ then directly implies $\inter{T}=\inter{\D}$.

Last, the boundedness of $\D$ is a consequence of~\Cref{r:boundedness}.
\end{proof}

%% file: 170-annex-decision.tex

\begin{lemma}\label{r:small-2-loop-can-simplify}
  If a \longStoC{} $\T$ satisfies the \smalltwoloop{}
  then the \Call{simplify}{} procedure is well-defined for arbitrarily-long histories in $\T^2$.
\end{lemma}

\begin{proof}
  From an arbitrary history in $\T^2$,
    it is always possible to find lassos
    for which both the initial part and the loop part have lengths less than $|Q|^2$.
\end{proof}


We can then prove that
  if a \longStoC{} satisfies the \smalltwoloop{}
  then the function it realizes satisfies the contextual Lipschitz property.


\begin{lemma}\label{r:bounded-fork}
  If a \longStoC{} $\T$ satisfies the \smalltwoloop{}
  then for all runs
    $\rho_1: \ttrans{}{c_1} i_1 \ttrans{u}{d_1} q_1$
    and $\rho_2: \ttrans{}{c_2} i_2 \ttrans{u}{d_2} q_2$,
    with $i_1,i_2$ initial states,
  we have $\distf(d_1 c_1 [\eps], d_2 c_2 [\eps]) \leq 10 \base |Q|^{|Q|+2}$.
\end{lemma}

\begin{proof}
  Let $H$ be the history containing only the two runs $\rho_1$ and $\rho_2$,
    and let $(\qB, c) = \Call{simplify}{\iB \act H}$.
  By \Cref{r:corr-simplify-s&c} and \Cref{r:small-2-loop-can-simplify},
    we have that $\tinit \act H [\eps] = \Delta_{\iB \act H} [\eps] = \Delta_{\qB} c [\eps]$.
  Let $e_1,e_2 \in \outC$ are such that $\Delta_{\qB} = \{ (q_1,e_1),(q_2,e_2) \}$.
  We thus have $d_1 c_1 [\eps] = e_1 c [\eps]$ and $d_2 c_2 [\eps] = e_2 c [\eps]$.
  Therefore $\distf(d_1 c_1 [\eps], d_2 c_2 [\eps]) \leq |e_1| + |e_2|$.
  Let $\qB = (x_{\qB}, \Delta, H_{\qB})$.
  By definition, we know that $\Delta_{\qB} = \Delta \act H_{\qB}$
    and that, by \Cref{r:boundedness},
      for all $(q,d) \in \Delta$, $|d| \leq 4 \base |Q|^{|Q|+2}$,
      and $|word(H_{\qB})| < \QQ$.
  Thus $|e_1| \leq 5 \base |Q|^{|Q|+2}$ and  $|e_2| \leq 5 \base |Q|^{|Q|+2}$
    and we obtain the result.
\end{proof}


\begin{proof}[Proof of \Cref{r:small-2-loop-implies-lip}]
  Let $\T$ be a \longStoC{}.
  Assume that $\T$ satisfies the \smalltwoloop{}
    and let $u_1,u_2 \in \dom(\inter{\T})$.
  We want to prove that there exists $K \in \N$ such that
    $\distf(\inter{\T}(u_1), \inter{\T}(u_2)) \leq K \dist(u_1, u_2)$.
  If $u_1 = u_2$ then $\distf(\inter{\T}(u_1), \inter{\T}(u_2)) = 0$,
    and the result is trivially obtained, whatever the value of $K$ is.

  In the following, we assume that $u_1 \neq u_2$ and thus $\dist(u_1, u_2) \geq 1$.
  Let $\rho_1: \ttrans{}{c_1} i_1 \ttrans{u_1}{d_1} f_1 \ttrans{}{e_1}$
  and $\rho_2: \ttrans{}{c_2} i_2 \ttrans{u_2}{d_2} f_2 \ttrans{}{e_2}$
    be the corresponding runs in $\T$.
  Let $u = \lcp(u_1, u_2)$ and $u_1',u_2'$ such that $u_1 = u u_1'$ and $u_2 = u u_2'$.
  Let $p_1,p_2 \in Q$ the states that $\rho_1$ and $\rho_2$ reach after having read $u$.
  That is $\rho_1: \ttrans{}{c_1} i_1 \ttrans{u}{d_1'} p_1 \ttrans{u_1'}{d_1''} f_1 \ttrans{}{e_1}$
    and $\rho_2: \ttrans{}{c_2} i_2 \ttrans{u}{d_2'} p_2 \ttrans{u_2'}{d_2''} f_2 \ttrans{}{e_2}$.

  By \Cref{r:bounded-fork},
    $\distf(d_1' c_1 [\eps], d_2' c_2 [\eps]) \leq 10 \base |Q|^{|Q|+2}$.
  Therefore,
    \begin{align*}
        \distf(\inter{\T}(u_1), \inter{\T}(u_2))
          &= \distf(e_1 d_1 c_1 [\eps], e_2 d_2 c_2 [\eps]) \\
          &\leq 10 \base |Q|^{|Q|+2} + |e_1 d_1''| + |e_2 d_2''| \\
          &\leq 10 \base |Q|^{|Q|+2} + \base (|u_1'| + 1) + \base (|u_2'| + 1) \\
          &\leq \base (10 |Q|^{|Q|+2} + \dist(u_1, u_2) + 2) \\
          &\leq \base (10 |Q|^{|Q|+2} + 3) \dist(u_1, u_2) \qedhere
    \end{align*}
\end{proof}

\begin{proof}[Proof of \Cref{r:decision-det}]
  By \Cref{t:main} and~\Cref{r:small-2-loop-implies-lip}, $\T$ admits an
  equivalent sequential \StoC{} transducer iff $\T$ satisfies the
  \smalltwoloop{} (see also the figure below).
%
  Because of this equivalence, we give a procedure to decide
    whether $\T$ satisfies the
  \smalltwoloop{}.

  The procedure first non-deterministically guesses a counter-example to the \smalltwoloop{}
    and then verifies that it is indeed a counter-example.
  By definition of the \smalltwoloop{},
    the counter-example can have one of the following four shapes:
  \begin{enumerate}
    \item\label{ctp-shape-1} a run
      $H :\; \ttrans{}{(c_0,d_0)} (i_1,i_2)
          \ttrans{u_1}{(c_1,d_1)} (p_1,p_2) \ttrans{u_2}{(p_1,p_2)} \pB$ in $\T^2$,
        with $|u_1| < |Q|^2$ and $|u_2| < |Q|^2$,
        that is a productive lasso neither commuting nor aligned.
    \item\label{ctp-shape-2} a run
      $H :\; \ttrans{}{(c_0,d_0)} (i_1,i_2)
          \ttrans{u_1}{(c_1,d_1)} (p_1,p_2) \ttrans{u_2}{(c_2,d_2)} (p_1,p_2)
          \ttrans{u_3}{(c_3,d_3)} (q_1,q_2) \ttrans{u_4}{(c_4,d_4)} (q_1,q_2)$ in $\T^2$,
        with $|u_i| < |Q|^2$, for all $i \in \{ 1, \dots, 4 \}$,
        such that the first lasso is a productive $\Com{x}$ lasso, for some $x \in \outLp$,
          and the second lasso is a productive lasso neither $\SCom{x}$ nor strongly aligned.
    \item\label{ctp-shape-3} a run
      $H :\; \ttrans{}{(c_0,d_0)} (i_1,i_2)
          \ttrans{u_1}{(c_1,d_1)} (p_1,p_2) \ttrans{u_2}{(c_2,d_2)} (p_1,p_2)$ in $\T^2$,
        with $|u_1| < |Q|^2$ and $|u_2| < |Q|^2$,
        that is a productive aligned lasso but,
          for $\Delta$ appropriately obtained with $\split_{nc}$,
          $\Lc{\T_\Delta}$ and/or $\Rc{\T_\Delta}$ do not satisfy the twinning property.
    \item\label{ctp-shape-4} a run
      $H :\; \ttrans{}{(c_0,d_0)} (i_1,i_2)
          \ttrans{u_1}{(c_1,d_1)} (p_1,p_2) \ttrans{u_2}{(c_2,d_2)} (p_1,p_2)
          \ttrans{u_3}{(c_3,d_3)} (q_1,q_2) \ttrans{u_4}{(c_4,d_4)} (q_1,q_2)$ in $\T^2$,
        with $|u_i| < |Q|^2$, for all $i \in \{ 1, \dots, 4 \}$,
        such that the first lasso is a productive aligned lasso,
          the second lasso is productive and strongly aligned but,
            for $\Delta$ appropriately obtained with $\extract_{nc}$,
            $\Lc{\T_\Delta}$ and/or $\Rc{\T_\Delta}$ do not satisfy the twinning property.
  \end{enumerate}

  Verifying that a lasso in $\T^2$ is not commuting (resp. not aligned) boils down to checking
    whether there exists no $x \in \outLp$
    such that the lasso is $\Com{x}$ (resp. no $f \in \outC$ and $w \in \outL$
    such that the lasso is $\Align{f,w}$).
  %
  %
  In both cases, the search space for the words $x,w$ and context $f$
    can be narrowed down to factors of the output contexts of the given lasso.
  Thus the verification for shape~\ref{ctp-shape-1} can be done in polynomial time.
  Similarly,
    the verification for shapes~\ref{ctp-shape-2}, \ref{ctp-shape-3} and~\ref{ctp-shape-3}
    can be done in polynomial time.
  Furthermore, all three shapes are of polynomial size, by definition of the
  \smalltwoloop{}, yielding the result.
  %
\end{proof}